\documentclass[twocolumn,twoside]{IEEEtran}
\usepackage{times} 
\usepackage{epsfig}
\usepackage{graphicx}
\usepackage{epstopdf}
\usepackage{color,subfigure, latexsym, amsmath, amsfonts, amssymb,cite,soul}
\usepackage{algorithm}
\usepackage{algorithmic}
\usepackage{epsfig}
\usepackage{graphicx}
\usepackage{epstopdf}
\usepackage{bm,amsmath,amssymb,amsfonts,graphicx,epsfig,amsthm,color, xcolor}
\usepackage{bbold,dsfont}
\usepackage{float}
\soulregister{\cite}7
\usepackage[hyphens]{url}
\PassOptionsToPackage{bookmarks=false}{hyperref}

\usepackage[depth=-1]{bookmark}

\usepackage{booktabs}       
\usepackage{amsfonts}       
\usepackage{microtype}      

\usepackage{times}
\usepackage{graphicx} 

\usepackage{wrapfig}



\usepackage{accents}
\newtheorem{proposition}{Proposition}
\newtheorem{lemma}{Lemma}

\newtheorem{theorem}{Theorem}
\newtheorem{definition}{Definition}
\newtheorem{assumption}{Assumption}
\newtheorem{problem}{Problem}
\theoremstyle{remark}\newtheorem{remark}{Remark}

\title{\LARGE \bf
	Learning Robust Data-based LQG Controllers from Noisy Data
}

\author{
	Wenjie~Liu,~Jian~Sun,~Gang Wang,~Francesco Bullo,~\IEEEmembership{Fellow,~IEEE},
	and Jie Chen,~\IEEEmembership{Fellow,~IEEE}
	\thanks{The work was supported in part by the National Natural Science Foundation of China under Grants 61925303, 62173034, 62088101, the BIT Research and Innovation Promoting Project under Grant
		2022YCXZ007, 
		and the Chongqing Natural Science Foundation under Grant 2021ZX4100027. 
		}
		\thanks{
			W. Liu, J. Sun, and G. Wang are with the National Key Lab of Autonomous Intelligent Unmanned Systems and  School of Automation, Beijing Institute of Technology, Beijing 100081, China, and also with the Beijing Institute of Technology Chongqing Innovation Center, Chongqing 401120, China (e-mail: liuwenjie@bit.edu.cn; sunjian@bit.edu.cn;  gangwang@bit.edu.cn).
			
			F. Bullo is with the Mechanical Engineering Department and the Center of Control, Dynamical Systems and Computation, UC Santa Barbara, CA 93106-5070, USA (e-mail: bullo@ucsb.edu).

			J. Chen is with the Department of Control Science and Engineering, Tongji University, Shanghai 201804, China, and also with the National Key Lab of Autonomous Intelligent Unmanned Systems, Beijing Institute of Technology, Beijing 100081, China 	
			(e-mail: chenjie@bit.edu.cn).
		}
	}

	\allowdisplaybreaks

\begin{document}
		\maketitle

		\begin{abstract}
			This paper addresses the joint state estimation and control problems for unknown linear
			time-invariant systems subject to both process and measurement noise. The aim is to redesign the
			linear quadratic Gaussian (LQG) controller based solely on data. The LQG controller comprises a
			linear quadratic regulator (LQR) and a steady-state Kalman observer; while the data-based LQR design
			problem has been previously studied, constructing the Kalman gain and the LQG controller from noisy
			data presents a novel challenge.
			
			In this work, a data-based formulation for computing the steady-state Kalman gain is proposed based
			on semi-definite programming (SDP) using some noise-free input-state-output data.
			Additionally, a data-based LQG controller is developed, which is shown to be equivalent to the
			model-based LQG controller. For cases where offline data are corrupted by noise, a robust data-based
			observer gain is constructed by tackling a relaxed SDP. The proposed controllers are proven to
			achieve robust global exponential stability (RGES) for state estimation and input-to-state practical
			stability (ISpS) under standard conditions. Finally, numerical tests are conducted to validate the
			proposed controllers' correctness and effectiveness.

		\end{abstract}

		\section{Introduction}\label{sec:intro}
		The states of a dynamical system refer to a set of variables that provide a complete characterization of its internal conditions or status at a given time \cite{simon2006optimal}, which plays a fundamental role in control, monitoring, and high-level information extraction of real-world systems.
		Unfortunately, the complete states are not always fully observable; and even if measurements can be collected, the inevitable sensor miscalibration and noise may degrade the accuracy. 
		As such, using the most likely value of a state reconstructed from raw measurements serves as an alternative, that is what is known as state estimation.
		Numerous efforts have contributed to the development of state estimation methods and theory as well as their applications in networked control systems, cyber-physical systems, and smart power grids; see e.g., \cite{catlin2012estimation,pasqualetti2012distributed,wang2019distribution,eng2022}.	

		Most existing results on state estimation capitalize on explicit system models describing the underlying dynamics.
		However, in many situations e.g., robotics \cite{DistCtrlRobotNetw}, biology \cite{kitano2002systems}, or human-in-the-loop systems \cite{inoue2019weak}, modeling using first-principles can be challenging or parameters obtained by system identification methods may be inaccurate or unreliable.
		On the contrary, data can be often collected through open-loop experiments.
		This approach has inspired the recent literature on data-driven control \cite{Hou2013from,persis2020data,pang2022comparison,markovsky2023data}, i.e., on designing controllers directly from data. 
		In the literature, most related results were derived based on the \emph{Willems et al.}'s fundamental lemma \cite{willems2005note}.
		Rather than capitalizing on an explicit system model, the lemma demonstrates that any finite-time trajectory of a linear time-invariant system can be expressed as a linear combination of some past trajectories generated by the system,  if the input signal is sufficiently rich.
		This offers an effective means for describing linear systems without explicit models.  
		Although this lemma has led to diverse data-driven control schemes, e.g., model predictive control \cite{berberich2019data,Coulson2019cdc,Liu2022data,sfl,csl2021cortes,muller2022data,breschi2022uncertainty} and explicit feedback control policies \cite{persis2020data,van2020data,li2022robust,kang2022minimum,wang2022event,wang2023data}, very few works have investigated its application in state estimation.
		
		In the context of data-driven state estimation, it is often assumed that input-state-output data can be collected offline while only input-output data are obtained during online operation.
		For noise-free systems, data-driven Luenberger observers were designed in \cite{van2020data,adachi2021dual} by exploiting the separation principle of estimation and control \cite{duncan1971on} together with the fundamental lemma.
		To handle unknown inputs or noise, unknown input observer (UIO) in \cite{turan2021data},  moving horizon estimation (MHE) in \cite{wolff2022robust}, and set-based estimation in \cite{alanwar2021data} were developed.
		In addition to bounding conditions, \cite{turan2021data} posed a more stringent condition on the noise dimension.
		Although \cite{wolff2022robust,alanwar2021data} only constrain the noise bound, an optimization problem should be solved online at each time, rendering them less computationally appealing for time-critical applications.
		It is worth emphasizing that, to estimate the state, these approaches require collecting the input data on-the-fly, as they do not tackle the controller design problem.
		
		The goal of this paper is to address robust data-based state estimation and control problems for unknown linear time-invariant systems subject to noise. The idea is using input-state-output data collected offline to redesign the linear quadratic Gaussian (LQG) controller
		that returns optimal state estimates and control inputs simultaneously under certain conditions.
		To fulfill this design, robust data-based formulations for the linear quadratic regulator (LQR) and steady-state Kalman observer are required.
		The problem of learning LQR from noisy data has been dealt with in \cite{depersis2021lowcomplexity}. 
		Our focus is thus placed on formulating the robust data-based Kalman observer and LQG controller. We consider two settings depending on whether pre-collected data are noise-corrupted or not.
		When noise-free data are collected offline, we propose a data-based SDP for computing
		the steady-state Kalman  gain.
		It is shown that the proposed data-based LQG controller is equivalent to its model-based counterpart provided the online noise is additive white Gaussian.
		When offline data are corrupted by noise, 
		we relax the original data-based SDP by adding a soft constraint to restore its feasibility as well as ensure the stability of the resulting Kalman gain.
		In this setting, although the controller does not share equivalence with the model-based LQG controller, it is robust to bounded noise.
		Finally, we establish, under rigorous conditions, robust global exponential stability (RGES) for the proposed estimator and input-to-state practical stability (ISpS) for the system in both cases.
		
		A few setups similar to ours appeared very recently in \cite{Alpago2020extended,al2022behavioral,furieri2022near,wang2022data}.
		To be specific, de-noising methods for offline data were investigated in \cite{Alpago2020extended,wang2022data}, to improve the performance of the data-driven predictive controller of \cite{Coulson2019cdc}.
		Precisely, the work of \cite{Alpago2020extended} designed a data-based extended Kalman filter for reducing the effect of noise in outputs, and a low-rank subspace identification noise estimation algorithm was developed in \cite{wang2022data}.
		By expressing a linear system equivalently in the behavior space, finding the LQG controller reduces to finding a static feedback controller which can be obtained using data \cite{al2022behavioral}.
		In \cite{furieri2022near}, a near-optimal safe output feedback controller was obtained using noisy data, whose suboptimality with respect to the model-based LQG controller was analyzed.
		It is worth stressing the differences between these works and ours.
		First, the results \cite{Alpago2020extended,al2022behavioral,furieri2022near,wang2022data} have only considered the data-driven control problem, whereas data-driven state estimation is investigated too in the present work.
		Moreover, \cite{al2022behavioral} requires the pre-collected data (a.k.a. expert data) be generated from the system that has been stabilized by an LQG controller, which is stringent and unrealistic.
		The work of \cite{furieri2022near} can only handle input and measurement noise but not the process noise.

		In a nutshell, the contributions of this paper are  as follows:
		\begin{itemize}
			\item [c1)]
			When noise-free data are collected, the data-based LQG controller is designed by solving data-based SDPs, which is shown equivalent to the model-based LQG controller; 
			\item [c2)]
			To deal with noise in offline data, robust data-based SDPs are derived by incorporating a soft constraint to recover the SDP feasibility and stability of the Kalman gain; and,
			\item [c3)]
			For both noise-free and noisy data-based controllers, it is shown that the estimator is RGES and the closed-loop system is ISpS under standard conditions.
		\end{itemize}

		\emph{Notation:}
		Denote the set of real numbers, integers, and positive integers by $\mathbb{R}$, $\mathbb{N}$, and $\mathbb{N}_{+}$, respectively.
		For a full row-rank matrix $M$, its right pseudo-inverse is denoted by $M^\dag$.
		Let $\| \cdot\|$ denote the vector Euclidean norm or  matrix spectral norm.
		Given a measurable time function $f : \mathbb{N} \rightarrow \mathbb{R}^n$ 
		and a time interval $[0,t)$, 
		its $\mathcal{L}_{\infty}$ norm  on $[0,t)$ is $\Vert f_t\Vert_{\infty} := {\rm ess}\sup_{s \in [0,t)} \Vert f(s)\Vert$.
		For matrices $A$, $B$, and $C$ with compatible dimensions, we abbreviate $ABC(AB)^\prime$ to $AB \cdot C[\star]^\prime$; $\lambda(A)$ represents the collection of all singular values (or eigenvalues if $A$ is square); $\underline{\lambda}_A$ ($\bar{\lambda}_A$)  the minimum (maximum) singular value of  $A$.
		Let $x_{[t_1, t_2]} := [x(t_1)~ x(t_1 + 1) \, \cdots \, x(t_2)]$ denote a stacked window of the signal $x$ in time interval $[t_1, t_2]$. 
		
		The Hankel matrix associated with the sequence $\{x(t)\}_{t = 0}^{N - 1}$, is denoted by
		\begin{equation*}
			H_{L}(x_{[0,N-1]}):=\left[
			\begin{matrix}
				x(0)& x(1) & \ldots & x(N-L)\\
				\vdots & \vdots & \ddots & \vdots \\
				x(L - 1) & x(L) & \ldots & x(N-1)
			\end{matrix}
			\right].
		\end{equation*}
		The definition of persistency of excitation is given below.
		\begin{definition}[{Persistency of excitation}]\label{def:pe}
			A sequence $\{x(t)\}_{t = 0}^{T - 1}$ with $x(t) \in \mathbb{R}^{n_x}$ and $T\ge (n + 1)L + 1$ is persistently exciting of order $L$ if ${\rm {rank}}(H_{L}(x_{[0,T-1]})) = n_x L$.
		\end{definition}
		\section{Preliminaries and Problem Formulation}\label{sec:preliminaries}
		This section introduces the problem of interest and some relevant results in the literature.
		
		\subsection{Problem formulation}\label{sec:preliminaries:problem}
		Consider the following linear system
		\begin{subequations}\label{eq:sys}
			\begin{align}
				x(t + 1) &= Ax(t) + Bu(t) + w(t)\label{eq:sys:x}\\ 
				y(t) &= Cx(t) + v(t)\label{eq:sys:y}
			\end{align}
		\end{subequations}
		where $x(t)\in \mathbb{R}^{n_x}$ is the state , $u(t)\in \mathbb{R}^{n_u}$ the control input, $y(t)\in \mathbb{R}^{n_y}$  the output, $w(t)\in \mathbb{R}^{n_x}$  the process noise, and $v(t)\in \mathbb{R}^{n_y}$ the measurement noise.
		Throughout this paper, the following assumptions hold.
		\begin{assumption}[{Controllability and observability}]\label{as:co}
			The pair $(A, B)$ is controllable and the pair $(A,C)$ is observable.
		\end{assumption}
		\begin{assumption}[System trajectory]\label{as:iso-io}
			The matrices $(A, B, C)$ in \eqref{eq:sys} are unknown.
			Input-state-output trajectories, i.e., $x_{[0,T]}$, $u_{[0,T - 1]}$ and $y_{[0,T - 1]}$ of system \eqref{eq:sys} can be collected during offline experiments, and the input sequence $u_{[0,T]}$ is persistently exciting of order $n_x + 1$.
			During online operation, only input-output data, i.e., $u(t)$ and $y(t)$, are available.
		\end{assumption}
		\begin{assumption}[Bounded noise]\label{as:noise}
			For all $t \in \mathbb{N}$, there exist non-negative constants $\bar{w}$ and $\bar{v}$ such that $w(t) \in \mathbb{B}_{\bar{w}}:= \{w|\Vert w \Vert \le \bar{w}\}$ and $v(t) \in \mathbb{B}_{\bar{v}} := \{v|\Vert v \Vert \le \bar{v}\}$.
		\end{assumption}

		\begin{remark}
			\label{rmk:as}
			Assumption \ref{as:iso-io} can be fulfilled in a remote networked estimation and control scenario.
			Specifically, if the state is of large size (e.g., video), transmitting via a network can be time-consuming and costly for real-time applications.
			In addition, for states containing sensitive information, eavesdropping or damage may happen during network transmission. 
			As such, transmitting outputs and recovering the state afterwards appears an appealing alternative.
			On the other hand, state can be collected in offline experiments and transmitted only once to the controller using possibly a different physical medium.
			This assumption has appeared in recent data-driven state estimation results  \cite{turan2021data,wolff2022robust}.
			For the data-driven set-based estimation method in \cite{alanwar2021data}, matrix $C$ was assumed known, which is similar to the assumption of having partial knowledge of the states during online inference. 
			Instead of Assumption \ref{as:noise}, the results of \cite{turan2021data,wolff2022robust} impose additional constraints on the noise.
			In \cite{turan2021data}, it is required that the dimensions of noise vectors $w(t)$ and $v(t)$ should not exceed the dimension of the output; see \cite[Proposition 1]{valcher1999state}.
		In \cite{wolff2022robust}, no process noise is considered. Instead of \eqref{eq:sys:x}, they deal with noise-corrupted versions of the state, i.e., $\tilde{x}(t) = x(t) + w(t)$ from the noise-free system
		$x(t + 1) = Ax(t) + Bu(t)$.
		In comparison, system \eqref{eq:sys:x} with process noise is more natural and challenging	 as discussed in \cite[Remark 5]{wolff2022robust}.
	\end{remark}
	
	Existing results on data-driven state estimation have focused on autonomous systems or systems with known inputs.
	For open-loop unstable systems, a stabilizing controller should be separately designed and equipped which complicates the architecture and implementation.
	To bypass this drawback, we are interested in addressing the joint state estimation and control problems of system \eqref{eq:sys} under Assumptions \ref{as:co}---\ref{as:noise}.
	To this end, some standard definitions on the desired performance are given.
	Dealing with noisy data, rather than exponential stability, we consider RGES and ISpS.
	The former is extended from \cite[Definition 9]{luo2016robust} and has been considered in \cite{wolff2022robust}.
	\begin{definition}[RGES of an estimator]\label{def:rges}
		Consider system \eqref{eq:sys} with bounded noise $w(t) \in \mathbb{B}_{\bar{w}}$ and $v(t) \in \mathbb{B}_{\bar{v}}$.
		A state estimator $\hat{x}(t)$ is RGES if there exist functions $\phi \in \mathcal{KL}$ and $\pi_{w}, \pi_v \in \mathcal{K}$\footnote{
			A function $\pi : [0,\infty) \rightarrow [0, \infty)$ is a $\mathcal{K}$-function if it is continuous and strictly increasing satisfying $\pi(0) = 0$.
			A function $\phi : [0,\infty)\times [0,\infty) \rightarrow [0, \infty)$ is a $\mathcal{KL}$-function if $\phi(\cdot, t)$ is of class $\mathcal{K}$ for each $t \in \mathbb{N}$ and $\phi(s, t)$ decreases to $0$ as $t \rightarrow \infty$ for any $s \in \mathbb{N}$.}
		such that for any $x(0), \hat{x}(0)\in \mathbb{R}^{n_x}$ and $w(t) \in \mathbb{B}_{\bar{v}}$, $v(t) \in \mathbb{B}_{\bar{v}}$, the following is satisfied for all $t\in \mathbb{N}$
		\begin{align}\label{eq:prges}
			\Vert x(t) - \hat{x}(t)\Vert &\le \phi(\Vert x(0) - \hat{x}(0)\Vert, t) \nonumber\\
			&\quad + \pi_w(\Vert w_t\Vert_{\infty})+ \pi_v(\Vert v_t\Vert_{\infty}).
		\end{align}
	\end{definition}
	
	The following definition of ISpS for discrete-time systems is adapted from \cite[Definition 2.2]{Z1994Small}.
	\begin{definition}[ISpS]
		\label{def:iss}
		A control system $x(t+1) = f(x(t), w(t))$ is called input-to-state practically stable (ISpS) if there exist $\mathcal{KL}$-function $\phi$, $\mathcal{K}$-function $\pi_w$, and constant $D_0$, such that for each $x(0) \in \mathbb{R}^{n_x}$ and each measurable essentially bounded input $w(t)$ on $ [0, +\infty)$, the solution $x(t)$ satisfies
		\begin{equation*}
			\Vert x(t)\Vert \le \phi(\Vert x(0)\Vert,t) + \pi_w(\Vert w_t\Vert_{\infty}) + D_0,\quad \forall t \in \mathbb{N}.
		\end{equation*}
		
	\end{definition}
	
	
	In this setting, we consider the following problem.
	\begin{problem}\label{problem0}
		Consider the system \eqref{eq:sys} under Assumptions \ref{as:co}---\ref{as:noise}.
		Design a data-based state estimator $\hat{x}(t)$ achieving RGES and a data-driven controller $u(t)$ ensuring system \eqref{eq:sys} ISpS.
	\end{problem}
	
	\subsection{Model-based SDP for LQG control}\label{sec:preliminaries:lqg} 
	
	When matrices $A$, $B$ and $C$ in \eqref{eq:sys} are known, Problem \ref{problem0} can be addressed by the LQG controller.
	In the following, we briefly review its model-based design, which plays an important role in deriving the main results of this paper.
	
	Under Assumption \ref{as:co}, a dynamic measurement feedback controller stabilizing system \eqref{eq:sys} exists and is given by \cite{OreillyJ}
	\begin{subequations}\label{eq:oc}
		\begin{align}
			\hat{x}(t + 1) &= (A + BK - LC)\hat{x}(t) + Ly(t)\\
			u(t) &= K\hat{x}(t)
		\end{align}
	\end{subequations}
	where matrices $K$ and $L$ are chosen such that $A + BK$ and $A - LC$ are Schur stable.
	The LQG controller which provides optimal control inputs while returning optimal state estimates is a special case of the controller \eqref{eq:oc}, when i) $w(t), v(t)$ are  mutually independent additive white Gaussian noise (AWGN) processes obeying $\mathbb{E}[w(t)w^\prime(t)] = N_x$ and $\mathbb{E}[v(t)v^\prime(t)] = N_y$ for all $t\in \mathbb{N}$, and ii) in addition to Schur stable, the gain matrices $K$ and $L$ are designed such that the quadratic cost is minimized. 
	
	%
	Precisely, considering weight matrices $W_x >  0$ and $W_u> 0$, matrix $K$ is chosen to minimize the cost function $J  = \lim_{T \rightarrow \infty} \mathbb{E}[(1/T)[\sum_{t = 0}^{T - 1}x(t) W_x x'(t) + u(t)W_u u'(t)]$, which is the classical stochastic LQR problem.
	This has been shown in \cite[Section 6.4]{chen2012optimal} equivalent to the deterministic $\mathcal{H}_2$ problem, i.e., finding $K$ such that $A + BK$ is Schur stable and ${\rm tr}(W_x P) + {\rm tr}(W_u KPK' )$ is minimized where $P$ is the controllability Gramian of system \eqref{eq:sys:x}.
	Such a controller gain matrix $K$ is unique and can be found as follows
	\begin{equation}\label{eq:lqr:K}
		\bar{K} = -(W_u + B'PB)^{-1}B'PA
	\end{equation}
	where $P>0$ is the unique solution to the following discrete-time algebraic Riccati (DARE) equation
	\begin{equation}\label{eq:ric:lqr}
		A'PA - P - A'PB(W_u + B'PB)^{-1}B'PA + W_x = 0.
	\end{equation}
	
	It is shown in \cite[Section 3.2]{depersis2021lowcomplexity} that the problem of finding $\bar{K}$ in \eqref{eq:lqr:K}
	can be equivalently posed as a semi-definite program
	\begin{align}\label{eq:lqr}
		& \min_{\gamma , K, P, G } \gamma  \\
		&~~~{\rm s.t.} ~\begin{cases}
			(A + BK)P(A + BK)^\prime - P + I \le 0\\
			P \ge I\\
			G  - KPK^\prime \ge 0\\
			{\rm tr}(W_x P) + {\rm tr}(W_u G ) \le \gamma .
		\end{cases}\nonumber
	\end{align}
	

	The matrix $L$ in the LQG controller is called the steady-state Kalman gain and is designed as follows.
	Let $e(t) = x(t) - \hat{x}(t)$ denote the estimation error  evolving as follows
	\begin{equation}\label{eq:est_error}
		e(t + 1) = (A - LC)e(t) + w(t) - Lv(t)
	\end{equation}
	and $L$ is chosen to minimize the trace of estimation error covariance matrix $\Sigma = \lim_{T \rightarrow \infty} \mathbb{E}[(1/T) \sum_{t = 0}^{T - 1}e(t)e'(t)]$, which is the classical Kalman filtering problem.
	According to \cite[Section 7.3]{goodwin2014adaptive}, the steady-state Kalman gain can be obtained as follows
	\begin{equation}\label{eq:lqg:L}
		\bar{L} = A \Sigma C'(C\Sigma C' + N_y)^{-1}
	\end{equation}
	where $\Sigma$ is the unique positive definite solution to the DARE equation
	\begin{equation}\label{eq:ric:lqg}
		A\Sigma A'  - \Sigma  + N_x - A\Sigma C'(C\Sigma C' + N_y)^{-1}C\Sigma A' = 0.
	\end{equation}
	It can be seen that, upon replacing $(A, B, W_x, W_u, K)$ by $(A', C', N_x, N_y, L')$, the DARE \eqref{eq:ric:lqg} boils down to \eqref{eq:ric:lqr}. 
	Therefore, the Kalman filtering problem can be formulated as the dual of LQR, i.e., finding $L$ such that  ${\rm tr}(N_x\Sigma ) + {\rm tr}(N_yL'\Sigma L)$ is minimized and $A - LC$ is Schur stable,
	where $\Sigma$ is the observability Gramian of   \eqref{eq:sys}.
	Hence, the problem of finding $\bar{L}$ in \eqref{eq:lqg:L} can be equivalently formulated as the following SDP
	\begin{align}\label{eq:lqg}
		& \min_{\epsilon, L, \Sigma , \Upsilon} \epsilon\\
		&~~~{\rm s.t.}~ \begin{cases}
			(A - LC)^\prime \Sigma (A - LC) - \Sigma  + I \le 0\\
			\Sigma  \ge I\\
			\Upsilon - L^\prime \Sigma  L \ge 0\\
			{\rm tr}(N_x\Sigma ) + {\rm tr}(N_y\Upsilon) \le \epsilon 
		\end{cases}\nonumber
	\end{align}
	which can be seen as the discrete-time counterpart of the Kalman filtering
	formulation for continuous-time systems in 
	\cite[Section 156.6]{matthew2021lmi}.
	
	In words, the model-based LQG  controller is expressed by
	\begin{subequations}\label{eq:oc:lqg}
		\begin{align}
			\hat{x}(t + 1) &= (A + B\bar{K} - \bar{L}C)\hat{x}(t) + \bar{L}y(t)\\
			u(t) &= \bar{K}\hat{x}(t)
		\end{align}
	\end{subequations}
	where the LQR controller gain $\bar{K}$ is given in \eqref{eq:lqr:K} and the steady-state Kalman gain $\bar{L}$ given in \eqref{eq:lqg:L}.
	The problem of finding the gain matrices can be converted into two deterministic $\mathcal{H}_2$ problems and addressed by solving SDPs \eqref{eq:lqr} and \eqref{eq:lqg}.
	Therefore, if we can redesign the controller \eqref{eq:oc:lqg} with matrices $\bar{K}$ and $\bar{L}$ by only using data, then Problem \ref{problem0} is solved.
	
	\subsection{Relevant data-driven results}
	\label{sec:preliminaries:dd}
	Data-based versions of the formulas above have been reported in the literature.
	For example, \cite[Section V]{van2020data} provided a data-based framework for representation and control of linear systems using dynamic output measurement feedback \eqref{eq:oc}.
	Data-driven LQR problem equivalent to \eqref{eq:lqr} was formulated in \cite{depersis2021lowcomplexity}.
	A brief review of related results is given as follows.
	
	Consider the $T$-long data $x_{[0,T - 1]}$, $x_{[1,T]}$, $u_{[0,T - 1]}$, $w_{[0,T - 1]}$, $y_{[0,T - 1]}$, $v_{[0,T - 1]}$ obtained from an open-loop experiment on system \eqref{eq:sys}.
	Rearrange these data into matrices, yielding
	\begin{subequations}
		\begin{align}\label{eq:data}
			& X_0 = [x(0)~x(1)~\cdots~x(T - 1)]\\
			& X_1 = [x(1)~x(2)~\cdots~x(T)]\\
			& U_0 = [u(0)~u(1)~\cdots~u(T - 1)]\\
			& W_0 = [w(0)~w(1)~\cdots~w(T - 1)]\\
			& Y_0 = [y(0)~y(1)~\cdots~y(T - 1)]\\
			& V_0 = [v(0)~v(1)~\cdots~v(T - 1)]
		\end{align}
	\end{subequations}
	which satisfy 
	\begin{equation}\label{eq:sys_data_repre:noise}
		\left[
		\begin{matrix}
			X_1\\Y_0
		\end{matrix}
		\right] = \left[
		\begin{matrix}
			A&B\\
			C&0
		\end{matrix}
		\right]\left[
		\begin{matrix}
			X_0\\
			U_0
		\end{matrix}
		\right] + 
		\left[
		\begin{matrix}
			W_0\\
			V_0
		\end{matrix}
		\right].
	\end{equation}
	When considering the noise-free case, i.e., $W_0 = 0$ and $V_0 = 0$, the equation reduces to
	\begin{equation}\label{eq:sys_data_repre:ideal}
		\left[
		\begin{matrix}
			X_1\\Y_0
		\end{matrix}
		\right] = \left[
		\begin{matrix}
			A&B\\
			C&0
		\end{matrix}
		\right]\left[
		\begin{matrix}
			X_0\\
			U_0
		\end{matrix}
		\right].
	\end{equation}
	Let $\Phi_0 := [X_0'~U_0']'$.
	A data-based version of controller \eqref{eq:oc} is constructed in \cite{van2020data}, which we outline as follows.
	\begin{theorem}
		[\!\!{\cite[Theorem 34]{van2020data}}]
		\label{thm:ddoc:ideal}
		Consider data $X_1$, $X_0$, $U_0$, $Y_0$ generated from system \eqref{eq:sys} with noise $w(t) =0$ and $v(t)= 0$, $\forall t \in \mathbb{N}$.
		Under Assumption \ref{as:co}, there exist matrices $\Phi_1^{\dag}$, $\Phi_2^{\dag}$, $K$, $L$ such that  $\Phi_0 [\Phi_1^{\dag}~\Phi_2^{\dag}]=\Phi_0\Phi_0^{\dag}= I$, and both $X_1(\Phi_1^{\dag}+\Phi_2^{\dag}K)$  and $(X_1 - LY_0)\Phi_1^{\dag}$ are Schur stable
		if and only if the following holds
		\begin{equation}\label{eq:rank:phi}
			{\rm rank}(\Phi_0) = n_x + n_u.
		\end{equation}
		A stabilizing measurement feedback controller of the form \eqref{eq:oc} is given by 
		\begin{subequations}\label{eq:ddoc:i}
			\begin{align}
				\hat{x}(t \!+\! 1) &\!=\! (X_1 \Phi_1^\dag \!+\! X_1 \Phi_2^\dag K \!-\! L Y_0\Phi_1^\dag)\hat{x}(t) \!+\! Ly(t)\label{eq:ddoc:i:est}\\
				u(t) &= K\hat{x}(t).\label{eq:ddoc:i:ctrl}
			\end{align}
		\end{subequations}
	\end{theorem}
	The condition \eqref{eq:rank:phi} is critical for guaranteeing existence of the controller \eqref{eq:ddoc:i}.
	When there is no noise in input-state data, i.e., $W_0 = 0$,  condition \eqref{eq:rank:phi} can be met by using a sufficiently exciting input sequence $u_{[0,T - 1]}$.
	\begin{lemma}[\!\!{\cite[Corollary 2]{willems2005note}}]\label{lem:pe}
		For controllable system \eqref{eq:sys} with $w(t) = 0$ for all $t \in \mathbb{N}$,
		if $u_{[0,T - 1]}$ is persistently exciting of order $n_x + 1$, then condition \eqref{eq:rank:phi} holds.
	\end{lemma}

	
	Having formulated the data-based version of controller \eqref{eq:oc}, it can be observed that if the gain matrices $\bar{K}$ in \eqref{eq:lqr:K} and $\bar{L}$ in $\eqref{eq:lqg:L}$ can also be derived using data, then a data-based LQG controller is constructed.
	
	Previous works have addressed the data-based LQR problem in both noise-free data \cite{persis2020data} and noisy data \cite{depersis2021lowcomplexity} scenarios.
	Specifically, if $W_0 = 0$, SDP \eqref{eq:lqr} can be equivalently reformulated using data.
	\begin{theorem}[\!{\cite[Theorem 1]{depersis2021lowcomplexity}}]\label{thm:sdp_lti_ideal}
		Let $U_0$, $X_0$ and $X_1$ be data collected from an experiment on system \eqref{eq:sys} with $W_0 = 0$ using a persistently exciting input sequence of order $n_x + 1$.
		Then the condition \eqref{eq:rank:phi} holds and the SDP 
		\begin{align}\label{eq:lqr_ideal}
			& \min_{\gamma , Q, P, G } \gamma \\
			&~~~  {\rm s.t.}~ \begin{cases}
				X_1 QP^{-1}Q^\prime X_1^\prime - P + I \le 0\\
				P \ge I\\
				G  -  U_0QP^{-1}Q^\prime U_0^\prime \ge 0\\
				X_0Q = P\\
				{\rm tr}(W_x P) + {\rm tr}(W_u G ) \le \gamma 
			\end{cases}\nonumber
		\end{align}
		is feasible.
		Moreover, let $(\gamma_{*}, Q_{*}, P_*,G_{*})$ be any optimal solution of \eqref{eq:lqr_ideal}, and $K_* = U_0Q_*P_*^{-1}$ recovers the LQR controller gain in \eqref{eq:lqr:K}, i.e., $\bar{K} = K_*$.
	\end{theorem}
	When $W_0 \ne 0$, Lemma \ref{lem:pe} does not hold in general and both the feasibility of SDP \eqref{eq:lqr_ideal} and the optimality of the controller gain are deteriorated.
	To mitigate this issue, a modified SDP was proposed in \cite[Section 5]{depersis2021lowcomplexity}, which introduces a soft constraint to balance between robustness against noise and optimality.
	\begin{theorem}[\!\!{\cite[Section 5]{depersis2021lowcomplexity}}]\label{thm:sdp_lti_noise}
		Let $U_0$, $X_0$, $X_1$ and $W_0$ be data collected from an experiment on system \eqref{eq:sys} using a persistently exciting input sequence of order $n_x + 1$.
		Consider the SDP 
		\begin{align}\label{eq:lqr_noise}
			& \min_{\gamma , Q, P, G , M_1} \gamma \\
			&~~~~~{\rm s.t.}~ ~~\begin{cases}
				X_1 QP^{-1}Q^\prime X_1^\prime - P + I \le 0\\
				P \ge I\\
				G  -  U_0QP^{-1}Q^\prime U_0^\prime \ge 0\\
				M_1 - QP^{-1}Q^\prime \ge 0\\
				X_0Q = P\\
				{\rm tr}(W_x P) + {\rm tr}(W_u G ) + \alpha_1{
					\rm tr}(M_1) \le \gamma  
			\end{cases}\nonumber
		\end{align}
		where $\alpha_1 > 0$ is arbitrary. 
		Then there exists some $\delta_r >0$ such that if $\Vert W_0\Vert \le \delta_r$, then the condition \eqref{eq:rank:phi} is satisfied and SDP \eqref{eq:lqr_noise} is feasible.
		In addition, letting $(\tilde{\gamma}, \tilde{Q}, \tilde{P}, \tilde{G},\tilde{M}_1)$ be any optimal solution of SDP \eqref{eq:lqr_noise},
		then $u(k) = \tilde{K} x(k)$ with $\tilde{K} = U_0 \tilde{Q} \tilde{P}^{-1}$ is stabilizing.
	\end{theorem}
	
	Bearing in mind the results in Subsections \ref{sec:preliminaries:lqg} and \ref{sec:preliminaries:dd}, Problem \ref{problem0} can be rephrased as follows.
	\begin{problem}\label{problem}
		For system \eqref{eq:sys} obeying Assumptions \ref{as:co}---\ref{as:noise},
		design a data-based observer and controller to fulfill the  LQG controller \eqref{eq:oc:lqg}, which achieves RGES for the data-driven state estimator and ISpS for the system.
	\end{problem}
	
	We will focus on Problem \ref{problem} in the rest of the paper.
	
	\section{Data-based LQG Control Using Noise-free Data}
	\label{sec:lqg:i}
	
	In this section, we address Problem \ref{problem} by considering that noise-free input-state-output data can be collected offline.
	\begin{assumption}[{Noise-free offline data}]
		\label{as:noise-free}
		Let $X_1$, $X_0$, $U_0$ and $Y_0$ be data collected from an offline experiment on system \eqref{eq:sys} in the absence of noise, i.e., $W_0 = 0$ and $V_0 = 0$.
	\end{assumption} 
	
	Under this assumption and invoking Theorem \ref{thm:sdp_lti_ideal}, the matrix $\bar{K}$ in \eqref{eq:lqr:K} can be obtained by resorting to the optimal solution of SDP \eqref{eq:lqr_ideal}.
	Building on this result, a possible solution for Problem \ref{problem} reduces to deriving a data-based version of  $\bar{L}$ in \eqref{eq:lqg:L}.
	Thus, this section proceeds by first constructing a steady-state Kalman gain $\bar{L}$ based on data.
	Substituting the data-based formulations for designing the
	gain matrices $\bar{K}$ and $\bar{L}$ into \eqref{eq:ddoc:i}, we show that  \eqref{eq:ddoc:i} is equivalent to \eqref{eq:oc:lqg}.
	Finally, theoretical guarantees on RGES and ISpS are provided.
	%
	
	Noticing from Assumption \ref{as:iso-io}, input sequence $u_{[0,T - 1]}$ is persistently exciting of order $n_x + 1$.
	According to Lemma \ref{lem:pe}, the condition \eqref{eq:rank:phi} holds and there exists  matrix $\Phi_0^{\dag} := [\Phi_1^\dag~\Phi_2^\dag]$ such that $\Phi_0 \Phi_0^{\dag} = I$.
	In addition, under Assumption \ref{as:noise-free}, matrices $X_1$, $X_0$, $U_0$ and $Y_0$ adhere to \eqref{eq:sys_data_repre:ideal}. 
	Hence, post-multiplying  both sides of  \eqref{eq:sys_data_repre:ideal} by $\Phi_0^{\dag}$ returns
	\begin{equation}\label{eq:abcddata}
		A = X_1 \Phi_1^\dag,~ B = X_1\Phi_2^\dag,~ C = Y_0\Phi_1^\dag.
	\end{equation}
	Substituting \eqref{eq:abcddata} into \eqref{eq:lqg}, we arrive at the following SDP
	\begin{align}\label{eq:lqg_ideal}
		& \min_{\epsilon, \Pi, \Sigma , \Upsilon} \epsilon\\
		&~~~{\rm s.t.}~ \begin{cases}
			\Phi_1^{\dag \prime}(\Sigma X_1 + \Pi Y_0)^\prime \Sigma ^{-1}(\Sigma X_1 + \Pi Y_0)\Phi_1^\dag - \Sigma  + I \le 0\\
			\Sigma  \ge I\\
			\Upsilon - \Pi^\prime \Sigma ^{-1}\Pi \ge 0\\
			{\rm tr}(N_x\Sigma ) + {\rm tr}(N_y\Upsilon) \le \epsilon .
		\end{cases}\nonumber
	\end{align}
	The equivalence between SDPs \eqref{eq:lqg} and \eqref{eq:lqg_ideal} is formally stated as follows. 
	\begin{theorem}
		[Equivalence between \eqref{eq:lqg} and \eqref{eq:lqg_ideal}]
		\label{thm:lqg:ideal}
		Under Assumptions \ref{as:co}, \ref{as:iso-io} and \ref{as:noise-free}, SDP \eqref{eq:lqg_ideal} is feasible.
		Moreover, any optimal solution $(\epsilon_*, \Pi_*, \Sigma_*, \Upsilon_*)$ of \eqref{eq:lqg_ideal} is such that $L_* = \Sigma_*^{-1}\Pi_*$ is equivalent to the optimal solution $\bar{L}$ obtained from \eqref{eq:lqg}.
	\end{theorem}
	\begin{proof}
		It follows from Assumption \ref{as:iso-io} and Lemma \ref{lem:pe} that condition \eqref{eq:rank:phi} holds, and hence there exist matrices $\Phi_1^\dag$, $\Phi_2^\dag$ such that $\Phi_0[\Phi_1^\dag \Phi_2^\dag] = I$ and SDP \eqref{eq:lqg} is feasible.
		Let $(\bar{\epsilon}, \bar{L}, \bar{\Sigma}, \bar{\Upsilon})$ be an optimal solution of \eqref{eq:lqg} and $(\epsilon_*, \Pi_*, \Sigma_*, \Upsilon_*)$ be an optimal solution of \eqref{eq:lqg_ideal}.
		The key idea of the  proof is to show that a candidate solution of \eqref{eq:lqg} can be constructed by any optimal solution of \eqref{eq:lqg_ideal}, and vice versa.
		Let $(\epsilon, L, \Sigma, \Upsilon) := (\epsilon_*, \Sigma_*^{-1}\Pi_*, \Sigma_* \Upsilon_*)$.
		It follows from \eqref{eq:sys_data_repre:ideal} that $X_1\Phi_1^{\dag} = A$ and $Y_0\Phi_1^{\dag} = C$, hence we have that $X_1 \Phi_1^{\dag} - \Sigma^{-1}\Pi Y_0 \Phi_1^{\dag} = A- LC$.
		This implies that $(L, \Sigma, \Upsilon)$ satisfies the first three constraints in \eqref{eq:lqg}.
		In addition, noticing that $\epsilon = \epsilon_* = {\rm tr}(N_x\Sigma_*) + {\rm tr}(N_y\Upsilon_*)$, one has that $(\epsilon_*, \Sigma_*^{-1}\Pi_*, \Sigma_* \Upsilon_*)$ is a feasible solution of problem \eqref{eq:lqg} and $\epsilon_* \ge \bar{\epsilon} =  {\rm tr}(N_x\bar{\Sigma}) + {\rm tr}(N_y\bar{\Upsilon})$.
		Accordingly, let $\hat{\Pi} = \bar{\Sigma} \bar{L}$, and thus $(\hat{\epsilon}, \hat{\Pi}, \hat{\Sigma}, \hat{\Upsilon}) := (\bar{\epsilon}, \bar{\Sigma} \bar{L}, \bar{\Sigma}, \bar{\Upsilon})$ is a feasible solution to \eqref{eq:lqg_ideal}.
		This implies that $\bar{\epsilon} \ge \epsilon_*$, and consequently, $\bar{\epsilon} = \epsilon_*$.
		This shows that $\bar{L} = \Sigma_*^{-1}\Pi_*$, which completes the proof.
	\end{proof}
	\begin{remark}[Influence of $\Phi_1^\dag$]
		To solve \eqref{eq:lqg_ideal}, the pseudo-inverse $\Phi_1^{\dag}$ of matrix $X_0$ should be computed first.
		Different choices of $\Phi_1^{\dag}$ may have an influence on the error between the optimal observer gain obtained using the model-based SDP \eqref{eq:lqg} and that using the data-based SDP \eqref{eq:lqg_ideal}.
	\end{remark}
	
	Notice from Subsection \ref{sec:preliminaries:dd} that the LQR controller gain $\bar{K}$ can be obtained by solving the data-based SDP in \eqref{eq:lqr_ideal}.
	In other words, the LQR controller gain satisfies $\bar{K} = K_* = U_0Q_*P^{-1}_*$ for any optimal solution of \eqref{eq:lqr_ideal}. 
	Substituting matrices $K_*$ and $L_*$ derived from data-based SDPs into \eqref{eq:ddoc:i}, a data-driven LQG controller resembling the model-based one in \eqref{eq:oc:lqg} is given by
	\begin{subequations}\label{eq:ddoc}
		\begin{align}
			\hat{x}(t \!+\! 1) &= (X_1 \Phi_1^\dag \!+\! X_1 \Phi_2^\dag K_* \!-\! L_* Y_0\Phi_1^\dag)\hat{x}(t) \!+\! L_*y(t)\label{eq:ddoc:est}\\
			u(t) &= K_*\hat{x}(t)\label{eq:ddoc:ctrl}
		\end{align}
	\end{subequations}
	where $\hat{x}(t)$ is the estimated state from the data-driven estimator \eqref{eq:ddoc:est} with $\hat{x}(0) = 0$ and $u(t)$ is the data-driven control input based on the estimated state.

	\begin{theorem}[Equivalence between \eqref{eq:oc:lqg} and \eqref{eq:ddoc}]
		\label{thm:oc:ideal}
		Consider 
		system \eqref{eq:sys} 
		with  mutually independent 
		AWGN processes $w(t)$ and $v(t)$  obeying $\mathbb{E}[w(t)w^\prime(t)] = N_x$ and $\mathbb{E}[v(t)v^\prime(t)] = N_y$, $\forall t$. 
		Under  Assumptions \ref{as:co}, \ref{as:iso-io} and \ref{as:noise-free}, condition \eqref{eq:rank:phi} is met, and SDPs \eqref{eq:lqr_ideal} and  \eqref{eq:lqg_ideal} are feasible.
		Let $(\gamma_*, Q_*, P_*, G_*)$ and $(\epsilon_*, \Pi_*, \Sigma_*, \Upsilon_*)$ denote respectively the unique optimal solution of SDP \eqref{eq:lqr_ideal} and \eqref{eq:lqg_ideal}.
		Then the data-based LQG controller \eqref{eq:ddoc} with $K_* = U_0 Q_*P_*^{-1}$ and $L_* = \Sigma_*^{-1} \Pi_*$ is equivalent to the model-based LQG controller \eqref{eq:lqg} in terms of the control input $u(t)$ and state estimate $\hat{x}(t)$.
	\end{theorem}

	\begin{proof}
		It has been shown in Theorem \ref{thm:sdp_lti_ideal} that \eqref{eq:lqr_ideal} is equivalent to \eqref{eq:lqr}.
		In addition, Theorem \ref{thm:lqg:ideal} indicates that \eqref{eq:lqg_ideal} is equivalent to \eqref{eq:lqg}. 
		It follows from \eqref{eq:sys_data_repre:ideal} that \eqref{eq:ddoc} is equivalent to \eqref{eq:oc:lqg}, which completes the proof.
	\end{proof}
	
	\begin{remark}[\emph{LQG control under bounded noise}]
		In the linear-quadratic-Gaussian case, the LQG controller gives the optimal control input. 
		Yet, the motivation of this paper is to fill in the gap in the literature of data-driven state estimation for open-loop unstable systems.
		Hence, we follow tradition by considering bounded noise in Assumption \ref{as:noise}.
		With slight abuse of terminology, we still refer to
		the controller \eqref{eq:ddoc} as LQG controller in the rest of the paper.
	\end{remark}

	Building on Theorems \ref{thm:lqg:ideal} and \ref{thm:oc:ideal}, we are ready to present the main result of this section.
	\begin{theorem}[RGES and ISpS of \eqref{eq:ddoc}]
		Under Assumptions \ref{as:co}---\ref{as:noise-free},  
		condition \eqref{eq:rank:phi} is satisfied, and SDPs \eqref{eq:lqr_ideal} and \eqref{eq:lqg_ideal} are feasible.
		Then the state estimator \eqref{eq:ddoc:est} is RGES and the system \eqref{eq:sys} with $u(t)$ given by \eqref{eq:ddoc:ctrl} treating $w(t)$ as input is ISpS.
	\end{theorem}
	
	\begin{proof}
		According to the first constraint in \eqref{eq:lqr_ideal} and \eqref{eq:lqg_ideal}, both matrices $X_1 \Phi_1^\dag + X_1 \Phi_2^\dag K_*$ and $X_1 \Phi_1^\dag - L_* Y_0\Phi_1^\dag$ are Schur stable. Therefore, there exist constants $c_1 >1$, $c_2 > 1$, $\beta_1 \in (0,1)$ and $\beta_2 \in (0,1)$ such that $\Vert (X_1 \Phi_1^\dag + X_1 \Phi_2^\dag K_*)^t\Vert \le c_1 \beta_1^t$ and $\Vert(X_1 \Phi_1^\dag - L_* Y_0\Phi_1^\dag)^t \Vert \le c_2 \beta_2^t$.
		Recalling from Theorem \ref{thm:oc:ideal} that SDP \eqref{eq:lqg_ideal} is equivalent to \eqref{eq:lqg}.
		Hence, the estimation error $e(t) = x(t) - \hat{x}(t)$ obeys \eqref{eq:est_error} and recursively
		\begin{align*}
			&\Vert  e(t)\Vert   = \Big\Vert  (A - L_*C)^tx(0) +\sum_{i = 0}^{t-1}(A - L_*C)^i\\
			& \times w(t-i-1) - L_*v(t-i-1)\Big\Vert \nonumber\\
			& \le c_2 \beta_2^t\Vert x(0) \Vert +\sum_{i = 0}^{t-1}c_2 \beta_2^i \Vert w(t - i - 1) - L_*v(t - i - 1)\Vert\nonumber.
		\end{align*}
		where $\sum_{i = 0}^{t-1}c_2 \beta_2^i \Vert w(t - i - 1) - L_*v(t - i - 1)\Vert$ is a $\mathcal{K}$-function of  $w_{[0,t - 1]}$ and $v_{[0,t - 1]}$.
		Therefore, according to Definition \ref{def:rges}, the state estimator \eqref{eq:ddoc} is RGES.
		
		Recalling from Assumption \ref{as:noise} that $v(t) \in \mathbb{B}_{\bar{v}}$, the estimation error can be further bounded by 
		\begin{equation}\label{eq:e_bound}
			\Vert e(t)\Vert \!\le\! c_2 \beta_2^t\Vert x(0) \Vert +\! \sum_{i = 0}^{t-1}c_2 \beta_2^i  (\Vert w(t-i-1)\Vert  + \bar{v} \Vert L_*\Vert ).
		\end{equation}
		Bearing this in mind, it follows from \eqref{eq:sys} that
		\begin{align*}
			\Vert x(t)\Vert & = \Vert(A + BK_*)x(t-1) - BK_*e(t-1) + w(t-1)\Vert\\
			& = \Vert (A + BK_*)^tx(0) - \sum_{i = 0}^{t - 1}(A + BK_*)^i\\
			&~~~\times[BK_*e(t - i - 1) + w(t - i - 1)] \Vert \\
			&\overset{\eqref{eq:e_bound}}{\le} c_1\beta_1^t \Vert x(0)\Vert + \sum_{i = 0}^{t - 1}c_1\beta_1^i \Big[\Vert w_t\Vert_{\infty} + BK_*\\
			&~~~\times \Big(c_2\beta_2^{t - i + 1}\Vert x(0)\Vert +\!\sum_{j = 0}^{i - 1}c_2 \beta_2^j(\Vert w_i\Vert_{\infty} + \bar{v} \Vert L_*\Vert )\Big)\!\Big]\\
			& \le c_1 \beta_1^t\Vert x(0)\Vert + \sum_{i = 0}^{t - 1} c_3 \beta_3^{t - 1}\Vert x(0)\Vert \\
			&~~~+\sum_{i = 0}^{t - 1} c_1 \beta_1^i\Big[\Vert w_t\Vert_{\infty} + \sum_{j = 0}^{i - 1}c_4 \beta_2^j(\Vert w_i\Vert_{\infty} + \bar{v} \Vert L_*\Vert )\Big]
		\end{align*}
		where $\beta_3  := \max\{\beta_1, \beta_2\}$, $c_3 = c_1\Vert BK_*\Vert$ and $c_4 = c_2\Vert BK_*\Vert$.
		Noticing that there exist constants $\tilde{c}_3$ and $\tilde{\beta}_3 \in (\beta_3, 1)$ such that the second term in the last inequality satisfies $\sum_{i = 0}^{t - 1} c_3 \beta_3^{t - 1} \le tc_3 \beta_3^{t - 1} \le \tilde{c}_3 \tilde{\beta}_3^{t}$, yielding
		\begin{align*}
			\Vert x(t)\Vert &\le c_1\tilde{\beta}_3^t\Vert x(0)\Vert + \tilde{c}_3\tilde{\beta}_3^t\Vert x(0)\Vert + \frac{c_1(1 - \beta_1^t)}{1 - \beta_1}\Vert w_t\Vert_{\infty} \nonumber\\
			&~~~ + \frac{c_1 c_4(1 - (1-\beta_2)^t \beta_1^t)(\Vert w_t\Vert_{\infty} + \Vert L_*\Vert \bar{v})}{(1 - \beta_1+\beta_1\beta_2)(1 - \beta_2)} \\
			&\le c_5 \tilde{\beta}_3^t + \pi_w(\Vert w_t\Vert_{\infty}) + \frac{c_1 c_4\Vert L_*\Vert \bar{v}}{(1 -\! \beta_1+\beta_1\beta_2)(1 - \beta_2)}
		\end{align*}
		where $c_5 := c_1 + \tilde{c}_3$ and $\pi_w(\Vert w_t\Vert_{\infty}) := (c_4 + 1-\beta_2)c_1\Vert w_t\Vert_{\infty}/(1-\beta_1-\beta_2+\beta_1\beta_2)$ is a $\mathcal{K}$-function.
		According to Definition \ref{def:iss}, system \eqref{eq:sys} with the proposed LQG controller achieves ISpS, which completes the proof.
	\end{proof}

	
	\section{Robust Data-driven Control Using Noisy Data}
	
	The previous data-driven LQG controller requires clean offline data.
	This requirement can be infeasible in real systems and  limits practical use. 
	This motivates us to develop a robust  controller based on offline noisy  data.
	Relaxing the noiseless data Assumption \ref{as:noise-free}, Theorem \ref{thm:oc:ideal} does not hold.
	Our results in the following can be seen as robust versions of those in the previous section.
	Considering that a data-based SDP for computing robust controller gain $\tilde{K}$ has been suggested in Theorem \ref{thm:sdp_lti_noise}, this section addresses Problem \ref{problem} by modifying SDP \eqref{eq:lqg_ideal} to compensate for noise in the data, such that a robust stabilizing observer gain $\tilde{L}$ can be constructed based on its optimal solution.
	Finally, leveraging gain matrices $\tilde{K}$ and $\tilde{L}$, a robust version of controller \eqref{eq:ddoc_ideal} is devised for which stability analysis follows. 
	
	\subsection{Certainty-equivalent solution for \eqref{eq:lqg_ideal} under noisy data} 
	Due to noise, data matrices satisfy \eqref{eq:sys_data_repre:noise}. Consequently, if condition \eqref{eq:rank:phi} holds, $(A, B, C)$ can be represented by
	\begin{align*}
		A = (X_1 - W_0)\Phi_1^\dag, ~B = (X_1 - W_0)\Phi_2^\dag,~
		C = (Y_0 - V_0)\Phi_1^\dag.
	\end{align*}
	In this setting, 
	SDP \eqref{eq:lqg_ideal} is modified into
	\begin{align}\label{eq:lqg_i-n}
		& \min_{\epsilon, \Pi, \Sigma , \Upsilon} \ \epsilon\\
		&~~~{\rm s.t.}~ \begin{cases}
			\Phi_1^{\dag \prime }[\Sigma (X_1 \!-\! W_0)\!+\! \Pi(Y_0 \!-\! V_0)]^\prime \cdot \Sigma ^{-1}[\star]' \!-\! \Sigma  \!+\! I \!\le\! 0\\
			\Sigma  \ge I\\
			\Upsilon - \Pi^\prime \Sigma ^{-1}\Pi \ge 0\\
			{\rm tr}(N_x\Sigma ) + {\rm tr}(N_y\Upsilon) \le \epsilon \nonumber
		\end{cases}
	\end{align}
	whose feasibility is guaranteed by Theorem \ref{thm:lqg:ideal}.
	Denote its optimal solution by $(\epsilon_*, \Pi_*, \Sigma_*,\Upsilon_*)$ and compute the robust observer gain by $L_* = \Sigma_* ^{-1}\Pi_*$.
	Since the noise $W_0$ and $V_0$ are unknown, SDP \eqref{eq:lqg_i-n} cannot be solved directly.
	It can be observed that the difference between SDPs  \eqref{eq:lqg_ideal} and \eqref{eq:lqg_i-n} lies only in the first constraint.
	For convenience, define
	\begin{subequations}\label{eq:Theta_M_Psi}
		\begin{align}
			\Theta & := \Phi_1^{\dag \prime}(\Sigma X_1 + \Pi Y_0)^\prime \Sigma ^{-1}(\Sigma X_1 + \Pi Y_0)\Phi_1^\dag - \Sigma \label{eq:Theta}\\
			M & := W_0^\prime \Sigma  (X_1 - W_0) - X_1^\prime \Sigma  W_0 + W_0^\prime \Pi(Y_0 - V_0) \nonumber\\
			&~~~- X_1^\prime \Pi V_0 + Y_0^\prime \Pi^\prime W_0- V_0^\prime \Pi^\prime (X_1 - W_0) \nonumber \\
			&~~~ - Y_0^\prime \Pi^\prime \Sigma ^{-1}\Pi V_0- V_0^\prime \Pi^\prime \Sigma ^{-1}\Pi (Y_0 - V_0)\label{eq:M}\\
			\Psi & := \Phi_1^{\dag \prime} M \Phi_1^\dag\label{eq:Psi}
		\end{align}
	\end{subequations}
	which allow us to rewrite the first constraint of \eqref{eq:lqg_ideal} as $\Theta + I \le 0$ and that of \eqref{eq:lqg_i-n} as $\Theta + \Psi + I \le 0$.
	Since noise matrices $W_0$ and $V_0$ only show up in the first constraint of \eqref{eq:lqg_i-n}, i.e., in $\Psi$ in \eqref{eq:Psi}, it is natural to reformulate \eqref{eq:lqg_i-n} by discarding the noise term $\Psi$ for tractability, i.e., by using the same SDP as in the noise-free case 
	\begin{align}\label{eq:lqg_n}
		& \min_{\epsilon, \Pi, \Sigma , \Upsilon} \ \epsilon\\
		&~~~{\rm s.t.}~ \begin{cases}
			\Phi_1^{\dag \prime}(\Sigma X_1 + \Pi Y_0)^\prime \Sigma ^{-1}(\Sigma X_1 + \Pi Y_0)\Phi_1^\dag - \Sigma  + I \le 0\\
			\Sigma  \ge I\\
			\Upsilon - \Pi^\prime \Sigma ^{-1}\Pi \ge 0\\
			{\rm tr}(N_x\Sigma ) + {\rm tr}(N_y\Upsilon) \le \epsilon. 
		\end{cases}\nonumber
	\end{align}
	However, due to the noise-corrupted matrices $X_1$ and $Y_0$, two issues arise with \eqref{eq:lqg_n}.
	\begin{itemize}
		\item [i)]
		This SDP may not be feasible; 
		and,
		\item [ii)]
		Even if it is feasible, the corresponding observer gain $\hat{L}$ may not guarantee that $A - \hat{L}C$ is Schur stable.
	\end{itemize}
	
	To address these issues, sufficient conditions are provided in the following results, whose proof can be proceeded following \cite[Theorem 2]{depersis2021lowcomplexity}.
	
	Suppose that a solution $(\hat{\epsilon},\hat{\Pi},\hat{\Sigma}, \hat{\Upsilon})$ for \eqref{eq:lqg_n} is found and compute the observer gain by $\hat{L} = \hat{\Sigma}^{-1}\hat{\Pi}$.
	With a slight abuse of notation, let $(\epsilon_*, \Pi_*,\Sigma_*, \Upsilon_*)$ denote any optimal solution of \eqref{eq:lqg_i-n} and $L_* = \Sigma_*^{-1} \Pi_*$ the associated observer gain.
	Matrices $\hat{\Theta}, \,\hat{M}, \,\hat{\Psi}, \,\Theta_*, \, M_*$ and $\Psi_*$ are similarly defined as in \eqref{eq:Theta_M_Psi}.
	\begin{theorem}[Feasibility of SDP \eqref{eq:lqg_n}]\label{thm:lqg:i-n}
		Let Assumptions \ref{as:co}---\ref{as:noise} hold.
		For any $\eta_1  \ge 1$, there exists some $\delta \ge 0$ such that if $\max\{\Vert W_0\Vert, \Vert V_0 \Vert \} \le \delta$, then the condition \eqref{eq:rank:phi} is satisfied and the SDP \eqref{eq:lqg_n} is feasible with any optimal solution $(\hat{\epsilon},\hat{\Pi},\hat{\Sigma}, \hat{\Upsilon})$.
		If the following condition holds 
		\begin{equation}\label{eq:hat_Psi}
			\hat{\Psi} \le (1 - 1/\eta_1)I,
		\end{equation} 
		then the observer gain  $\hat{L} = \hat{\Sigma}^{-1}\hat{\Pi}$ is stabilizing.
	\end{theorem}
	

	Theorem \ref{thm:lqg:i-n} asserts that feasibility of SDP \eqref{eq:lqg_n} and stability of the associated observer gain $\hat{L}$ rely on two conditions: i)
	sufficiently small noise and ii)  condition \eqref{eq:hat_Psi}. 	
The former is standard in data-driven control, and it does not bring issues when replacing \eqref{eq:lqg_ideal} with \eqref{eq:lqg_n} under small-noise data.
	The challenge relates to how to ensure the condition \eqref{eq:hat_Psi}, where $\hat{\Psi}$ is the gap between the stability conditions of the ideal formulation \eqref{eq:lqg_i-n} and the certainty-equivalent \eqref{eq:lqg_n}.
	To fulfill \eqref{eq:hat_Psi}, we search for a solution $\hat{\Phi}$ with small norm.
	This requirement translates into requiring $\hat{M}$ to have small norm.
	However, there is no such constraint  in \eqref{eq:lqg_n}. Thus, even small noise may undermine the condition \eqref{eq:hat_Psi}.
	To tackle this issue, a new SDP is formulated by incorporating the robustness constraint \eqref{eq:hat_Psi} into \eqref{eq:lqg_n}.

	\subsection{Robust observer gain}
	As discussed above, we aim at designing a robust data-based SDP by incorporating the robustness constraint \eqref{eq:hat_Psi} into \eqref{eq:lqg_n}; that is, the solution of SDP automatically enforces the norm of $ \hat{\Psi}$.  
	Observe from \eqref{eq:Psi} that  $\Vert \hat{\Psi}\Vert \le \Vert \hat{M}\Vert \Vert \Phi_1^{\dag}\Vert^2 = \Vert \hat{M}\Vert\Vert \Phi_1^{\dag \prime} \Phi_1^\dag\Vert$, which implies that smaller $\Vert \hat{M}\Vert$ and $\Vert \Phi_1^{\dag \prime} \Phi_1^\dag\Vert$ leads to smaller $\Vert\hat{\Psi}\Vert$.
	Since $\Phi_1^\dag$ appears in \eqref{eq:lqg_n}, we obtain this matrix by solving an optimization problem minimizing $\Vert \Phi_1^{\dag \prime} \Phi_1^\dag\Vert$ as follows
	\begin{align}\label{eq:phi1}
		& \min_{\rho, M_2,\Phi_1^\dag, \Phi_2^\dag}  \ \rho\\
		&~~~~~{\rm s.t.}~ \ 
		\begin{cases}
			\Phi_0[\Phi_1^\dag~\Phi_2^\dag] = I\\
			\left[
			\begin{matrix}
				M_2 & \Phi_1^{\dag\prime}\\
				\Phi_1^\dag & I
			\end{matrix}
			\right] \ge 0\\
			{\rm tr}(M_2)  \le \rho.
		\end{cases}\nonumber
	\end{align}
	Feasibility of this SDP can be guaranteed under condition \eqref{eq:rank:phi}.
	Define its optimal solution by $(\tilde{\rho}, \tilde{M}_2, \tilde{\Phi}_1^{\dag}, \tilde{\Phi}_2^{\dag})$.
	Based on $\tilde{\Phi}_1^{\dag}$, a noise-robust version of SDP \eqref{eq:lqg_n} is given by
	\begin{align}\label{eq:lqg_noise}
		& \min_{\epsilon, \Pi, \Sigma , \Upsilon} \epsilon\\
		&~~~{\rm s.t.}\, \begin{cases}
			\tilde{\Phi}_1^{\dag \prime}(\Sigma X_1 + \Pi Y_0)^\prime \Sigma ^{-1}(\Sigma X_1 + \Pi Y_0)\tilde{\Phi}_1^{\dag} - \Sigma  + I \le 0\\
			\Sigma  \ge I\\
			\Upsilon - \Pi^\prime \Sigma ^{-1}\Pi \ge 0\\
			{\rm tr}(N_x\Sigma ) + {\rm tr}(N_y\Upsilon) + \alpha_2(\Vert \Upsilon\Vert + \Vert \Sigma \Vert + \Vert \Pi\Vert) \le \epsilon\nonumber
		\end{cases}
	\end{align}
	where $\alpha_2 >0$ is a parameter interpolating performance and robustness.
	Specifically, larger $\alpha_2$ values favor solutions with smaller $\Vert \Sigma \Vert + \Vert \Pi\Vert + \Vert \Upsilon\Vert$.
	Denote its optimal solution by $(\tilde{\epsilon}, \tilde{\Pi}, \tilde{\Sigma}, \tilde{\Upsilon})$.
	It follows from \eqref{eq:M} that 
	\begin{align}
		&\Vert \tilde{M}\Vert \le \max\big\{ \Vert W_0\Vert\Vert X_1 - W_0\Vert + \Vert  X_1\Vert\Vert W_0\Vert,~\Vert X_1\Vert\Vert V_0\Vert \nonumber\\
		&+\!\Vert W_0\Vert\Vert Y_0 \!-\! V_0\Vert \!+\! \Vert Y_0\Vert\Vert W_0\Vert \!+\! \Vert V_0\Vert\Vert X_1 \!-\! W_0\Vert, ~\Vert Y_0\Vert\Vert V_0\Vert\nonumber\\
		& + \Vert V_0\Vert\Vert Y_0 - V_0\Vert \big\}(\Vert \tilde{\Sigma} \Vert + \Vert \tilde{\Pi}\Vert + \Vert \tilde{\Pi}^\prime \tilde{\Sigma}  \tilde{\Pi}^{-1}\Vert)
	\end{align}
	which asserts that, the larger $\alpha_2$, the smaller $\Vert \tilde{M}\Vert$.
	A sufficient condition for \eqref{eq:hat_Psi} is 
	\begin{equation}
		\Vert \tilde{M}\Vert \Vert \tilde{\Phi}_1^{\dag \prime} \tilde{\Phi}_1^\dag \Vert \le 1 - 1/\eta_1.
	\end{equation}
	Obviously, larger $\alpha_2$ values promote more stabilizing solutions.
	Building on this observation, a counterpart of Theorem \ref{thm:lqg:i-n} for \eqref{eq:lqg_noise} is established.
	\begin{theorem}[Feasibility of SDP \eqref{eq:lqg_noise}]\label{thm:lqg_noise}
		Under Assumptions \ref{as:co}---\ref{as:noise}, 
		there exists a value $\delta_g \ge 0$ such that if $\max\{\Vert W_0\Vert, \Vert V_0 \Vert \} \le \delta_g$, then \eqref{eq:rank:phi} is satisfied and SDP \eqref{eq:lqg_noise} is feasible.
		In addition, the observer gain $\tilde{L} = \tilde{\Sigma}^{-1}\tilde{\Pi}$ is stabilizing.
	\end{theorem}
	
	The proof of Theorem \ref{thm:lqg_noise} is similar to that of \cite[Lemma 4]{depersis2021lowcomplexity}, and is thus omitted here due to space limitation.
	Having designed the gain matrices $\tilde{K}$ and $\tilde{L}$ based on offline noisy data, we proceed to construct a robust version of the controller \eqref{eq:ddoc_ideal}.

	\subsection{Robust data-based controller using noisy data}
	Leveraging  $(\tilde{\Phi}_1^{\dag},\tilde{\Phi}_2^{\dag})$ from solving \eqref{eq:phi1} and matrices $\tilde{K}$, $\tilde{L}$ constructed by optimal solutions of \eqref{eq:lqr_noise} and \eqref{eq:lqg_noise}, respectively, the controller \eqref{eq:ddoc} is robustified as follows
	\begin{subequations}\label{eq:ddoc_ideal}
		\begin{align}
			\hat{x}(t + 1) &= \big[(X_1 - W_0) \tilde{\Phi}_1^\dag + (X_1 - W_0) \tilde{\Phi}_2^\dag \tilde{K} \nonumber\\
			&~~~- \tilde{L} (Y_0 - V_0)\tilde{\Phi}_1^\dag\big]\hat{x}(t) + \tilde{L} y(t)\\
			u(t) &= \tilde{K} \hat{x}(t).
		\end{align}
	\end{subequations}
	Again, since noise $W_0$ and $V_0$ are not known, \eqref{eq:ddoc_ideal} cannot be implemented. 
	Upon ignoring  $W_0$ and $V_0$, we arrive at an implementable robust data-driven controller as follows
	\begin{subequations}\label{eq:ddoc_noise}
		\begin{align}
			\hat{x}(t + 1) &= (X_1 \tilde{\Phi}_1^\dag + X_1 \tilde{\Phi}_2^\dag \tilde{K} - \tilde{L} Y_0\tilde{\Phi}_1^\dag)\hat{x}(t) + \tilde{L} y(t)\label{eq:ddoc_noise:est}\\
			u(t) &= \tilde{K} \hat{x}(t)\label{eq:ddoc_noise:ctrl}.
		\end{align}
	\end{subequations}
	Stability analysis for the resulting system \eqref{eq:sys} with this controller is presented as follows.
	
	\newcounter{TempEqCnt} 
	\setcounter{TempEqCnt}{\value{equation}} 
	\setcounter{equation}{32} 
	\begin{figure*}[!h]
		\normalsize
	\begin{align}
		&\left[
		\begin{matrix}
			x(t + 1)\\
			\hat{x}(t + 1)
		\end{matrix}
		\right] = \left[
		\begin{matrix}
			(X_1 \!-\! W_0)\tilde{\Phi}_1^\dag\!\! & (X_1 - W_0)\tilde{\Phi}_2^\dag \tilde{K}\\
			\tilde{L}(Y_0 - V_0)\tilde{\Phi}_1^\dag\!\! & X_1\tilde{\Phi}_1^\dag \!+\! X_1\tilde{\Phi}_2^\dag \tilde{K}\! -\! \tilde{L} Y_0 \tilde{\Phi}_1^\dag
		\end{matrix}
		\right]
		\left[
		\begin{matrix}
			x(t)\\
			\hat{x}(t)
		\end{matrix}
		\right]+\left[
		\begin{matrix}
			w(t)\\
			\tilde{L} v(t)
		\end{matrix}
		\right] = \Xi_0	\left[
		\begin{matrix}
			x(t)\\
			\hat{x}(t)
		\end{matrix}
		\right]
		+ \left[
		\begin{matrix}
			w(t)\\
			\tilde{L} v(t)
		\end{matrix}
		\right] 
		\label{eq:extend_state:1}\\
		&\left[
		\begin{matrix}
			x(t + 1)\\
			e(t + 1)
		\end{matrix}
		\right] = \Xi_3	\left[
		\begin{matrix}
			x(t)\\
			e(t)
		\end{matrix}
		\right]
		+ \left[
		\begin{matrix}
			w(t)\\
			w(t) - \tilde{L} v(t)
		\end{matrix}
		\right],
		\label{eq:extend_state:2}\\
		&~ {\rm where}~~\Xi_3 = \underbrace{\left[
			\begin{matrix}
				(X_1 - W_0)\tilde{\Phi}_1^\dag + (X_1 - W_0)\tilde{\Phi}_2^\dag \tilde{K} \!\!&\!\! (X_1 - W_0)\tilde{\Phi}_2^\dag \tilde{K}\\
				0\!\! &\!\! X_1\tilde{\Phi}_1^\dag  - \tilde{L} Y_0 \tilde{\Phi}_1^\dag
			\end{matrix}
			\right]}_{\triangleq \Xi_1} 
		\!+ \!
		\underbrace{\left[
			\begin{matrix}
				0 \!\!&\!\! 0\\
				W_0 \tilde{\Phi}_1^\dag \!+\! W_0\tilde{\Phi}_2^\dag\tilde{K} \!-\! \tilde{L}V_0\tilde{\Phi}_1^\dag \!\!&\!\! W_0 \tilde{\Phi}_2^\dag \tilde{K}
			\end{matrix}
			\right]}_{\triangleq \Xi_2}
		\label{eq:extend_state:3}
	\end{align}
	
	\hrulefill
	\vspace*{4pt}
\end{figure*}
\setcounter{equation}{\value{TempEqCnt}} 

\begin{theorem}[RGES and ISpS for \eqref{eq:ddoc_noise}]
	Under Assumptions \ref{as:co}---\ref{as:noise}, 
	there exists  $\delta_{\bar{n}}>0$ such that for all $\bar{n} := \max\{\bar{w}, \bar{v}\} \le \delta_{\bar{n}}$ condition \eqref{eq:rank:phi} is satisfied and SDPs \eqref{eq:lqr_noise}, \eqref{eq:phi1} and \eqref{eq:lqg_noise} are feasible.
	Consider system \eqref{eq:sys} with controller \eqref{eq:ddoc_noise}, where $(\tilde{\Phi}_1^{\dag}, \tilde{\Phi}_2^{\dag})$ is an optimal solution of \eqref{eq:phi1}, and $\tilde{K} = U_0 \tilde{Q} \tilde{P}^{-1}$,  $\tilde{L} = \tilde{\Sigma}^{-1} \tilde{\Pi}$ 
	obtained by solving \eqref{eq:lqr_noise}
	and \eqref{eq:lqg_noise}, respectively.
	Then the state estimator \eqref{eq:ddoc_noise:est} is RGES and the system with $u(t)$ given by \eqref{eq:ddoc_noise:ctrl} treating $w(t)$ as input is ISpS.
\end{theorem}
\begin{proof}
	System \eqref{eq:sys} with controller \eqref{eq:ddoc_noise}  can be described by the composite system \eqref{eq:extend_state:1}, presented at the top of this page.
	Considering the estimation error  $e(t) = x(t) - \hat{x}(t)$, system  \eqref{eq:extend_state:1} can be transformed into the composite system \eqref{eq:extend_state:2}.
	Since the transformation from $[x(t)'~ \hat{x}'(t)]'$ to $[x(t)'~e'(t)]'$ is linear and non-singular, matrix $\Xi_3$ is similar to $\Xi_0$.
	From \cite[Section 6.6]{cullen2012matrices}, this implies that system \eqref{eq:extend_state:1} has same characteristic polynomial as system \eqref{eq:extend_state:2}. Hence, if $\Xi_3$ is Schur stable, so is $\Xi_0$. 
	According to \eqref{eq:extend_state:3}, one has that $\Vert \Xi_3\Vert = \Vert \Xi_1 + \Xi_2\Vert\le \Vert \Xi_1\Vert + \Vert \Xi_2\Vert$.
	In the following, we analyze $\|\Xi_1\|$ and $\|\Xi_2\|$ separately.
	
	The eigenvalues of matrix $\Xi_1$ can be expressed in terms of those of matrices $(X_1 - W_0)\tilde{\Phi}_1^\dag + (X_1 - W_0)\tilde{\Phi}_2^\dag \tilde{K}$ and $X_1\tilde{\Phi}_1^\dag - \tilde{L}Y_0 \tilde{\Phi}_1^\dag$, i.e., $\lambda(\Xi_1) = \lambda((X_1 - W_0)\tilde{\Phi}_1^\dag + (X_1 - W_0)\tilde{\Phi}_2^\dag \tilde{K}) \cup \lambda(X_1\tilde{\Phi}_1^\dag - \tilde{L}Y_0 \tilde{\Phi}_1^\dag)$ \cite[Section 1.6]{OreillyJ}.
	Theorems \ref{thm:sdp_lti_noise} and \ref{thm:lqg_noise} confirm that matrices $\tilde{K}$ and $\tilde{L}$ are stabilizable.
	Consequently,  $\Xi_1$ is Schur stable since system \eqref{eq:sys} is controllable and observable under Assumption \ref{as:co}.
	Since being Schur stable requires $\|\Xi_3\| <1$, this can be met if $\Vert \Xi_1\Vert + \Vert \Xi_2\Vert <  1$.
	Bearing this in mind, we  show in the following that if the noise is small enough, then $\Vert \Xi_1\Vert + \Vert \Xi_2\Vert <  1$ always holds.

	Let $n(t) := [w(t)'\ v(t)']'$.
	It follows that 
	\begin{align*}
		\Vert \Xi_2 \Vert \le\bar{n}\sqrt{T \!- 1}\big[\sqrt{n_x}(\Vert \tilde{\Phi}_0^\dag\Vert + \Vert \tilde{\Phi}_2^\dag\Vert \Vert \tilde{K}\Vert) + \sqrt{n_y}\Vert \tilde{L}\Vert\Vert \tilde{\Phi}_1^\dag\Vert\big]. 
	\end{align*}
	Hence, recalling $\delta_r$ defined in Theorem \ref{thm:sdp_lti_noise} and $\delta_g$  defined in Theorem \ref{thm:lqg_noise}, if $\bar{n} \le \delta_{\bar{n}} := \min\{\delta_r, \delta_g, \tilde{\delta}_{\bar{n}}\}$ with
	\begin{align*}
		\tilde{\delta}_{\bar{n}} \le \frac{1-\Vert\Xi_1 \Vert}{\sqrt{T - 1}\big[\sqrt{n_x}(\Vert \tilde{\Phi}_0^\dag\Vert + \Vert \tilde{\Phi}_2^\dag\Vert \Vert \tilde{K}\Vert) + \sqrt{n_y}\Vert \tilde{L}\Vert\Vert \tilde{\Phi}_1^\dag\Vert\big]}
	\end{align*}
	there exist a matrix $P_z \ge 0$ and a constant $\bar{\beta}_1 >0$ such that 
	\setcounter{equation}{35} 
	\begin{equation}
		(\Xi_1 + \Xi_2) P_z (\Xi_1 + \Xi_2)^\prime - P_z = \Xi_3 P_z \Xi_3^\prime - P_z\le -\bar{\beta}_1 I.
	\end{equation}
	Define the extended state by $z(t): = [x(t)' \  e(t)']'$.
	Choosing the Lyapunov function $V(z(t)) = z(t)^\prime P_z z(t)$, one gets that
	\begin{align*}
		& \ V(z(t\! +\!1)) -V(z(t))\\
		= &\, \bigg(\Xi_3 z(t) + \left[
		\begin{matrix}
			I&0\\
			I&-\tilde{L}
		\end{matrix}
		\right]n(t)\!\!\bigg)^\prime \!\cdot\! P_z [\star] - z(t)^\prime P_2 z(t)\\
		\le&\,-\big(\bar{\beta}_1 - \tau_1\Vert P_z\Vert^2 \Vert \Xi_3\Vert^2 \big) \Vert z(t)\Vert^2 \\
		&\,+ \Big(\frac{1}{\tau_1} + \Vert P_z\Vert \Big)(1 + \Vert \tilde{L}\Vert)^2\Vert n(t)\Vert^2\\
		\le&\, -\bar{\beta}_2 V(z(t)) + \tau_2 \Vert n(t)\Vert^2
	\end{align*}
	where $\tau_1 > 0$ is such that Young's inequality $[w(t)'~w(t)'-v(t)'\tilde{L}']P_z\Xi_3 z(t) \le \tau_1\Vert P_z\Vert^2 \Vert \Xi_3\Vert^2\Vert z(t)\Vert^2 + \Vert n(t)\Vert^2(1/\tau_1+\Vert P_z\Vert)$ holds.
	In addition, $\bar{\beta}_2 := \underline{\lambda}_{P_z}(\bar{\beta}_1 - \tau_1\Vert P_z\Vert^2\Vert \Xi_3\Vert^2) > 0$ and $\tau_2 := (1/\tau_1 + \Vert P_z\Vert)\Vert n(t)\Vert^2$.
	Recursively, it holds that
	\begin{equation*}
		V(z(t)) \le (1 - \bar{\beta}_2)^{t}V(z(0)) \!+ \tau_2\!\sum_{i = 0}^{t - 1}(1 - \bar{\beta}_2)^i\Vert n(t\!-i\!-1)\Vert^2.
	\end{equation*}
	Since $\underline{\lambda}_{P_z} \Vert e(t) \Vert^2 \le \underline{\lambda}_{P_z} \Vert z(t) \Vert^2 \le V(z(t)) \le \bar{\lambda}_{P_z} \Vert z(t)\Vert^2$ and $e(0) = x(0)$, the estimation error obeys
	\begin{align}
		\Vert e(t)\Vert &\le \sqrt{\frac{2\bar{\lambda}_{P_z}}{\underline{\lambda}_{P_z}}}(1 - \bar{\beta})^{\frac{t}{2}}\Vert e(0)\Vert \nonumber\\
		&~~~+
		\Big(\tau_2\sum_{i = 0}^{t - 1}(1 - \bar{\beta}_2)^i\Vert n(t-i-1)\Vert^2\Big)^{\frac{1}{2}}
	\end{align}
	which implies RGES for the estimator \eqref{eq:ddoc_noise:est}.
	Similarly, the state satisfies
	\begin{equation*}
		\Vert x(t)\Vert \le \sqrt{\frac{2\bar{\lambda}_{P_z}}{\underline{\lambda}_{P_z}}}(1 - \bar{\beta})^{\frac{t}{2}}\Vert x(0)\Vert +\tilde{\pi}_w(\Vert w_t\Vert_{\infty}) + \bar{v}\sqrt{\frac{2\tau_2}{\bar{\beta_2}}}
	\end{equation*}
	where $\tilde{\pi}_w(\Vert w_t\Vert_{\infty}) := \sqrt{2\tau_2\sum_{i = 0}^{t - 1}(1 - \bar{\beta}_2)^i}\Vert w_t\Vert_{\infty}$ is a $\mathcal{K}$-function of $\Vert w_t\Vert_{\infty}$.
	This establishes ISpS for the system according to Definition \ref{def:iss}. 
\end{proof}
\begin{remark}[\emph{Comparison with existing results}]
	Compared with the results in \cite{turan2021data,wolff2022robust,alanwar2021data}, the proposed method has five advantages.
	\begin{itemize}
		\item [1)]
		\emph{Less offline data.}
		The work \cite{turan2021data} assumes that the offline data $\{[u(i)^\prime~w(i)^\prime]\}_{i = 0}^{T - 1}$ collected 
		is persistently exciting of order $n_x + 1$; as such, at least $T \ge n_x + (n_u + n_x)(n_x + 1)$ data samples should be collected.
		In contrast, the proposed controller only requires $\{u(i) \}_{i = 0}^{T - 1}$ be persistently exciting of order $n_x + 1$, and $T \ge n_x + n_u(n_x + 1)$, which saves at least $n_x(n_x + 1)$ data samples.
		\item [2)]
		\emph{Less conservative assumptions on noise.}
		As discussed in Remark \ref{rmk:as}, in \cite{turan2021data} and \cite{wolff2022robust}, except for Assumption \ref{as:noise}, other constrains on noise are required.
		\item [3)]
		\emph{Less time and computing resources.}
		In \cite{wolff2022robust,alanwar2021data}, an optimization problem is solved at each time during online operation, which costs more time and consumes more computing resources than the proposed method. 
		\item [4)]
		\emph{Easy-to-check and less conservative conditions.}
		In \cite{turan2021data}, a kernel inclusion condition on the offline data is posed for stabilization, which is not easy to check in practice.
		In contrast, the stability conditions of our method are  Lyapunov-type and simple.
		Relative to the practical robust exponential stability in \cite{wolff2022robust}, we ensure RGES when both offline and online data are subject to noise.
		\item [5)]
		\emph{Tackling state estimation and control problems simultaneously.}
		In \cite{wolff2022robust,turan2021data}, both the input $u(k)$ and output $y(k)$ are used to estimate the state $x(k)$, and hence an additional controller should be implemented to generate $u(k)$.
		Our proposed controller estimates the state $x(k)$ requiring only the previous output $y(k - 1)$, and a simple feedback controller is designed based on the estimated state.
		In other words, the state estimation and control are jointly addressed involving a single dynamics, which simplifies  practical implementations.
	\end{itemize}
\end{remark}

\section{Numerical Examples}
In this section, we compare the proposed schemes with related results in \cite{turan2021data,alanwar2021data,wolff2022robust} through two numerical examples.
The following simulations were run on a Lenovo laptop with a 20-core i7-12700H processor at 2.3GHz.

\subsection{Noise-free offline data}
In the first experiment, we considered the discretized version of the open-loop unstable batch reactor system in \cite{persis2020data} using a sampling period of $0.1$s, with
matrices $A$, $B$ and $C$ in \eqref{eq:sys} given by
\begin{align*}
	&A = \left[
	\begin{matrix}
		1.178& 0.001& 0.511& -0.403\\
		-0.051& 0.661 &-0.011& 0.061\\
		0.076 &0.335& 0.560& 0.382\\
		0& 0.335& 0.089& 0.849
	\end{matrix}
	\right], \\
	&B = \left[
	\begin{matrix}
		0.004& -0.087\\
		0.467& 0.001\\
		0.213& -0.235\\
		0.213& -0.016
	\end{matrix}
	\right], ~~C= \left[
	\begin{matrix}
		1& 0& 1& -1\\
		0 &1& 0& 0
	\end{matrix}
	\right].
\end{align*}

All state estimation schemes were implemented in MATLAB.
We collected noise-free trajectories  of length $T = 15$
from random initial conditions with random inputs  uniformly generated from $[-1,1]$.
SDP \eqref{eq:lqg_ideal} was solved with $N_x = 0.02I_4$ and $N_y = 0.02I_2$ by CVX \cite{grant2014cvx}, yielding  
\begin{equation*}
	L_* = \left[
	\begin{matrix}
		0.7034 &   0.0385\\
		-0.0228 &   0.3812\\
		0.2753  &  0.4023\\
		0.0619  &  0.4066
	\end{matrix}
	\right]
\end{equation*}
which is the same as the observer $\bar{L}$ found using the MATLAB command ${\rm dare}(A',C',N_x, N_y)$ solving the DARE equation.
Moreover, noise $w(t)$ and $v(t)$ during online operation were generated uniformly from $[-0.02, 0.02]$ and the state and estimated state trajectories are depicted in Fig. \ref{fig:kalman1}, which verifies the effectiveness of the proposed method in Section \ref{sec:lqg:i}.

\begin{figure}
	\centering
	\includegraphics[width=8.2cm]{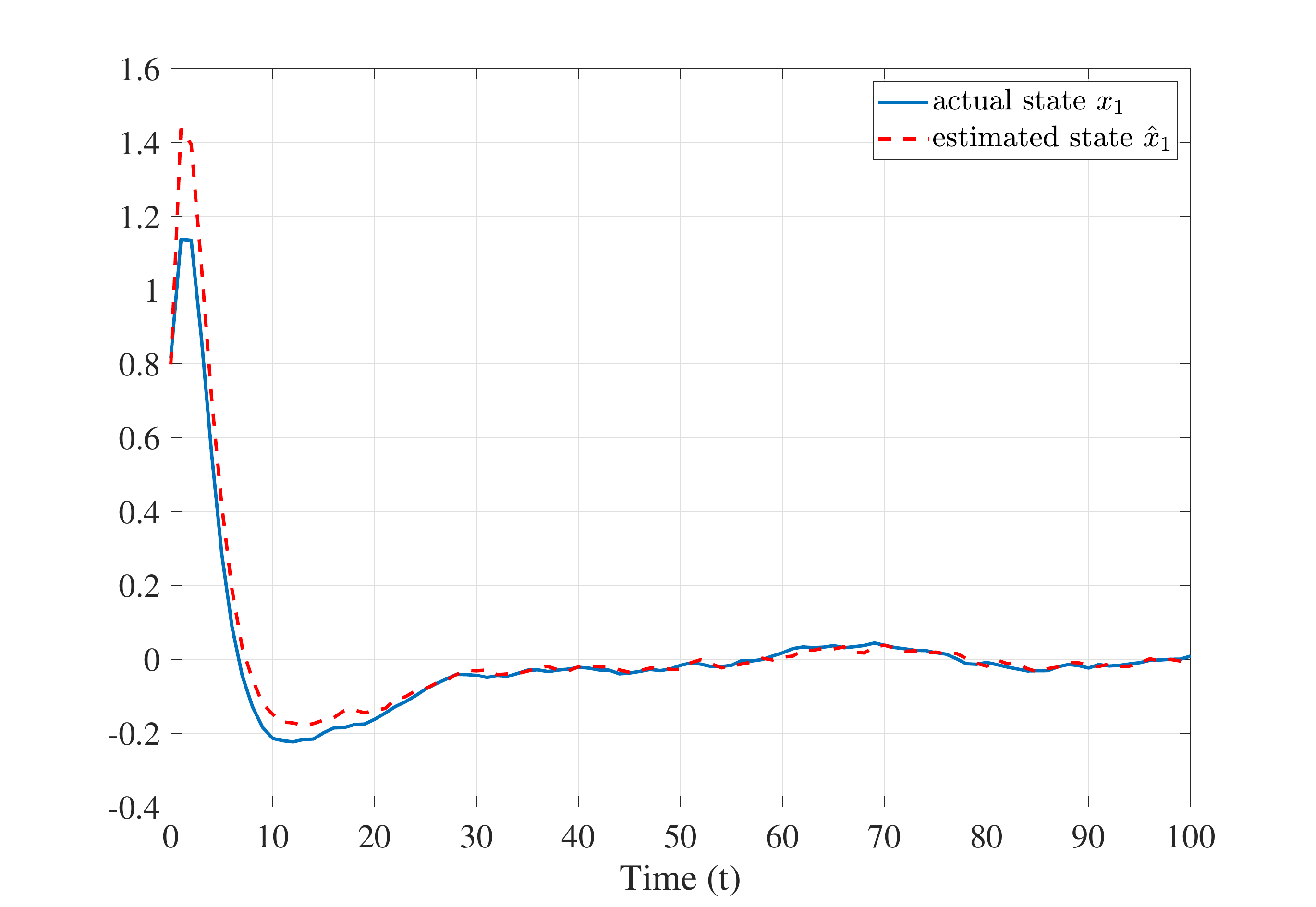}\\
	\caption{State trajectory by the controller \eqref{eq:ddoc}.}\label{fig:kalman1}
	\centering
\end{figure}

\subsection{Noisy offline data}
In the second experiment, we examine the performance of our proposed robust data-driven control method relative to three recent competing methods, including the data-based UIO in \cite{turan2021data}, data-based MHE in \cite{wolff2022robust}, and the data-driven set-based method in \cite{alanwar2021data} on the batch reactor system.
As discussed in Remark \ref{rmk:as}, the dimension of both process and measurement noise in \cite{turan2021data} should not exceed that of the output, i.e., \eqref{eq:sys:x} is replaced by $x(t + 1) =Ax(t)+ B u(t) + E_ww(t)$ with ${\rm rank}(E_w) \le n_y$. 
Therefore, we first compare the performance of the proposed method and the UIO method in the case of ${\rm rank}(E_w) = n_y =  2$, which is then followed by comparison of the four methods under ${\rm rank}(E_w) = n_y =  4$.
It is worth clarifying that the methods in \cite{turan2021data,wolff2022robust,alanwar2021data} were designed for autonomous systems, and therefore they were simulated with the input sequence generated from \eqref{eq:ddoc_noise}.

\begin{figure}
	\centering
	\includegraphics[width=8.2cm]{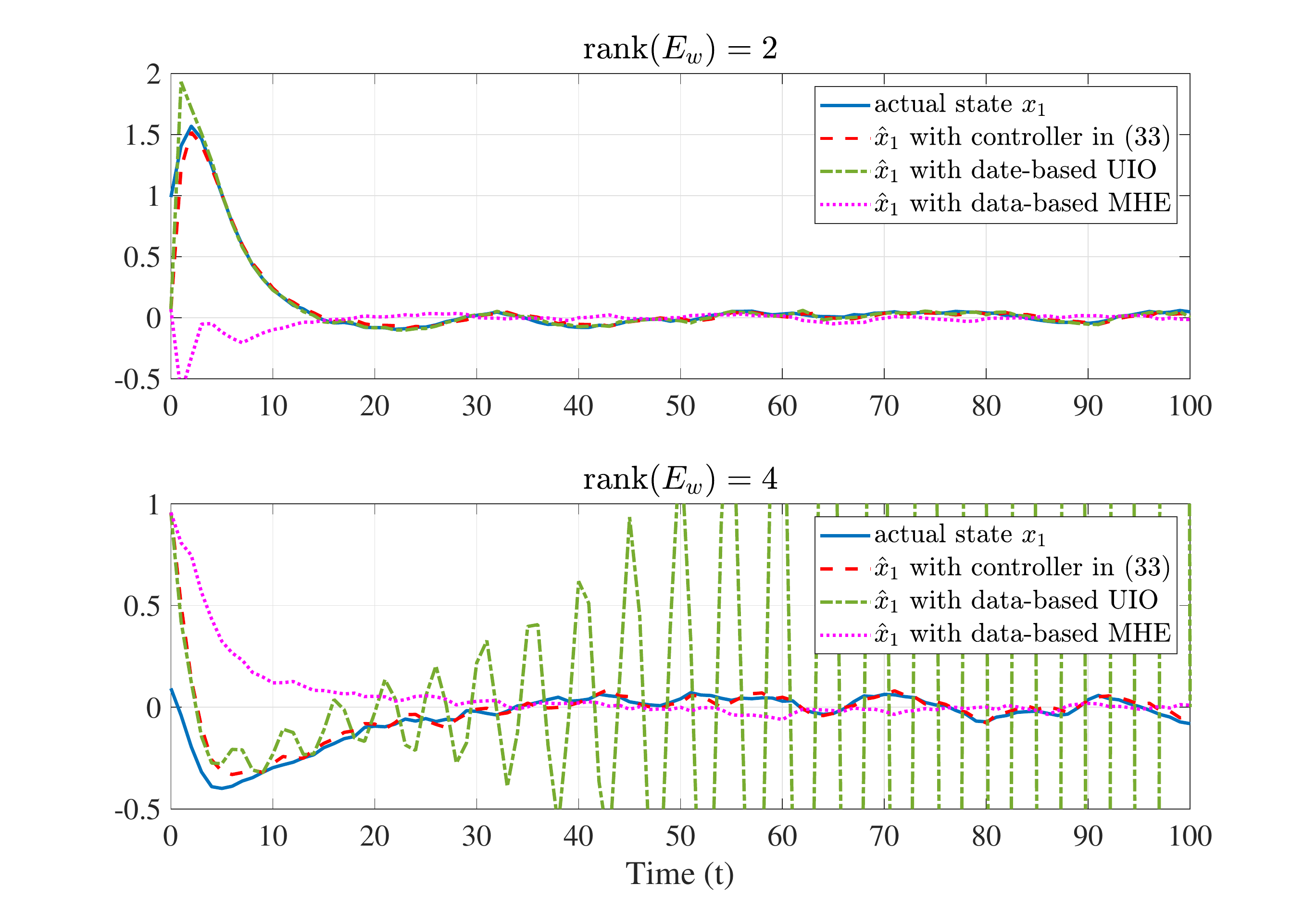}\\
	\caption{Performance of the proposed controller \eqref{eq:ddoc_noise}, data-based UIO state estimator \cite{turan2021data} and data-based MHE \cite{alanwar2021data} with $\bar{n} = 0.2$ under ${\rm rank} (E_w) = 2$ (top panel) and ${\rm rank} (E_w) = 4$ (bottom panel).
	}\label{fig:lqgvsuio}
	\centering
\end{figure}

Considering ${\rm rank}(E_w) = n_y =  2$, noise $w(t)$ and $v(t)$ were uniformly generated at random from $[-0.02, 0.02]$.
In addition SDP \eqref{eq:lqr_noise} was solved with $W_x = I_4$, $W_u = I_2$, $\alpha_1 = 0.2$ and SDP \eqref{eq:lqg_noise} was solved with $N_x = 0.02I_4$, $N_y = 0.02I_2$, $\alpha_2 = 0.2$. 
The offline data were collected following the same procedure as described in the previous experiment.
The top panel of Fig. \ref{fig:lqgvsuio} shows the state estimation performance of the proposed method \eqref{eq:ddoc_noise} (dashed line), the data-based UIO method in \cite{turan2021data} (dash-dotted line) and the data-based MHE method (dotted line) \cite{wolff2022robust}.
It can be observed that the state converges under the proposed robust controller with noisy data and the estimation error remains small.
The averaged state estimation error $\bar{e}: = (1/100)\sum_{t = 1}^{t = 100} \Vert \hat{x}(t) - x(t)\Vert$ of the proposed method is $0.489$, which is smaller than $0.545$ that of the data-based UIO and $1.815$ that of the data-based MHE.

Under ${\rm rank}(E_w) = 4$, we plot the estimated states by the three methods in the bottom panel of Fig. \ref{fig:lqgvsuio}.
Evidently, the estimated state by the proposed method remains convergent, while that by the UIO method diverges.
Moreover, the average estimation error of controller \eqref{eq:lqg_noise} $0.889$ is several times smaller than that of MHE method $2.311$.

\begin{figure}
	\centering
	\includegraphics[width=8.5cm]{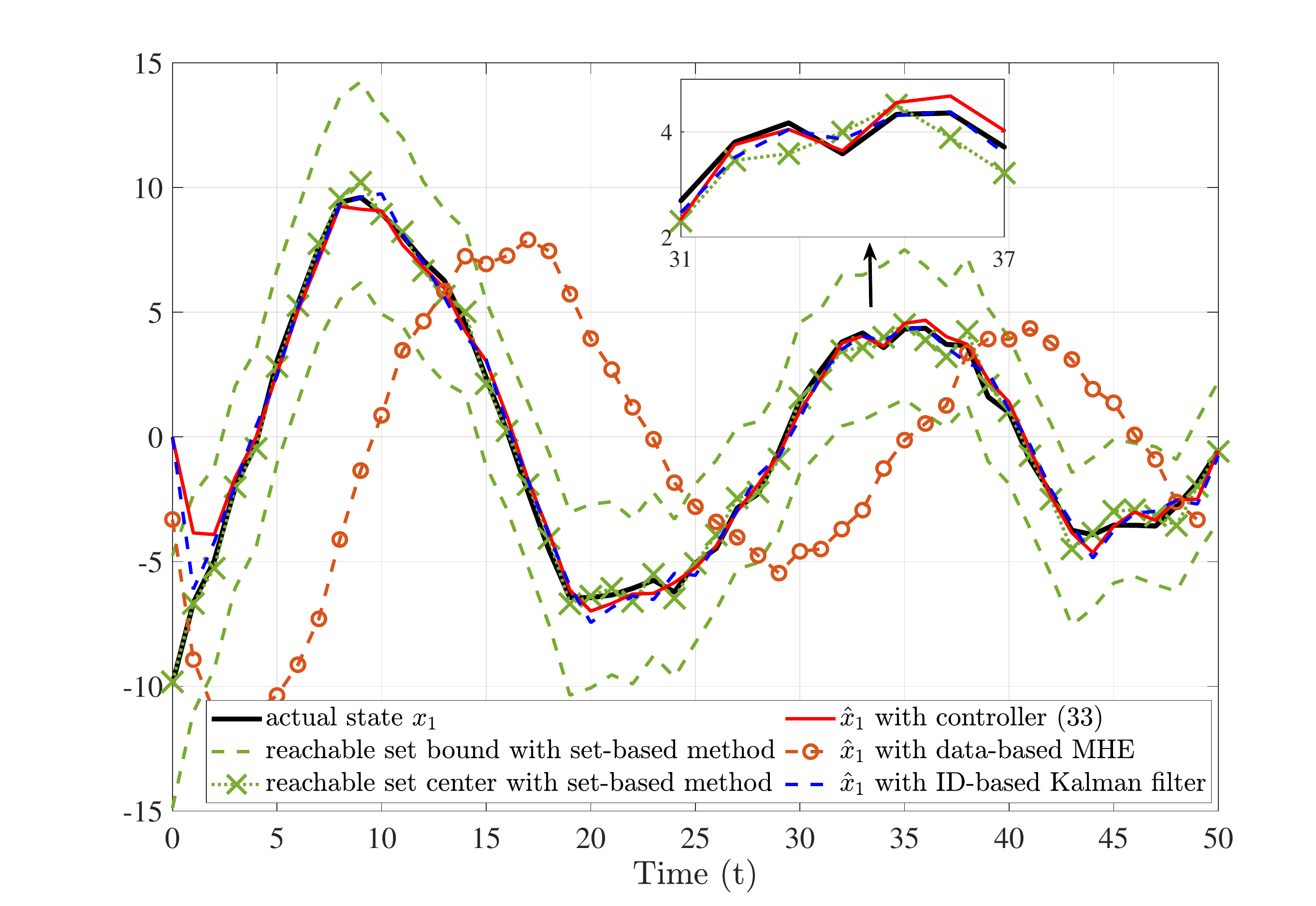}\\
	\caption{
		State estimation performance of the data-driven set-based method \cite{alanwar2021data},  data-based MHE  \cite{wolff2022robust},   ID-based Kalman filter, and proposed controller \eqref{eq:ddoc_noise}.
	}\label{fig:lqgvszono}
	\centering
\end{figure}

The working assumption of \cite{alanwar2021data} is slightly different from the others, which assumes the knowledge of matrix $C$.
To fulfill Assumption  \ref{as:iso-io}, an extension of this method is made, which is discussed in Appendix \ref{sec:appendix}.
The following simulation was proceeded by using the implicit intersection with constrained zonotopes in \cite{alanwar2021data}, which was shown to numerically outperform its baseline methods.
Moreover, the set-based method returns a bound on the reachable set of state $x(t)$ and we use its center as the estimated state.
Consider the same input-driven variant of the rotating target system as in \cite{alanwar2021data}
\begin{align*}
	A = \left[
	\begin{matrix}
		0.9455 \!&-0.2426\\
		0.2486  \! & 0.9455
	\end{matrix}
	\right], B = \left[
	\begin{matrix}
		0.1\\ 0
	\end{matrix}
	\right], C = \left[
	\begin{matrix}
		1&\! 0.4\\
		0.9 \!&-1.2\\
		-0.8\!& 0.2\\
		0 \!&0.7
	\end{matrix}
	\right] 
\end{align*}
which is an autonomous system and whose control inputs were randomly generated from $[-1,1]$.
Noise $w(t)$ and $v(t)$ were generated obeying a uniform distribution over $[-1, 1]$.
SDP \eqref{eq:lqg_noise} was solved with $N_x =I_4$, $N_y = I_2$ and $\alpha_2 = 1$.
Fig. \ref{fig:lqgvszono} compares the state estimation performance of the proposed method \eqref{eq:ddoc_noise} (red solid), the data-driven set-based method in \cite{alanwar2021data} using constraint zonotopes (green dotted with cross), the data-based MHE in \cite{wolff2022robust} (orange dashed with circle), and the identification (ID)-based Kalman method (blue dashed).
System ID was performed using the \emph{n4sid} toolbox with default option setting based on the same set of offline data.
Clearly, the proposed controller \eqref{eq:ddoc_noise} stabilizes the system.
Furthermore, the set-based method achieves better performance at the beginning, which is because $u(0)$ and $y(0)$ are used to estimate $x(0)$, yet $\hat{x}(0)$ is randomly generated for the proposed method.
Let $\bar{t}$ denote the average computation time per iteration.
Comparison of $\bar{e}$ and $\bar{t}$ between different methods is presented in Table \ref{tab:com}.
The proposed controller enjoys comparable state estimation performance while being (much) more computationally efficient that competing data-based methods. 

\begin{table}
	\centering
	\caption{Comparison between different methods with $\bar{n} = 1$}
	\label{tab:com}
	\renewcommand{\arraystretch}{1.2}
	\setlength{\tabcolsep}{10pt}
	\begin{tabular}{lcc}
		\toprule 
		Method & $\bar{e}$ & $\bar{t}$ (sec) \\
		\midrule
		Data-driven controller \eqref{eq:ddoc_noise}&$0.577$&$10^{-4}$\\
		Data-based	MHE in\cite{wolff2022robust}&$6.738$&$0.006$\\
		Data-driven set-based method in \cite{alanwar2021data}&$0.300$&$5.125$\\
		System ID-based Kalman filter&$0.608$&$0.005$\\
		\bottomrule
	\end{tabular}
\end{table}


\section{Conclusions}

This paper presented a solution to the problem of jointly
estimating and controlling the state of an unknown linear systems subject
to both process and measurement noise.  Specifically, this study proposed
data-based linear quadratic Gaussian (LQG) control formulations based on
offline data collection, both in noise-free and noisy settings.  For
noise-free offline data, a convex semi-definite programming (SDP) method is
used to construct a steady-state Kalman filter.  Building upon existing
research on data-driven linear quadratic regulator (LQR), this work also
presented a data-based version of the LQG controller that is equivalent to
its model-based counterpart.  For noisy data, a soft constraint is added to
the SDP to develop a robust SDP method that provides a robust observer
gain.  The proposed data-driven design enjoys robust global exponential
stability (RGES) of the state estimator and input-to-state practical
stability (ISpS) of the closed-loop system.  Numerical examples illustrate
the efficacy of the proposed controllers.


\section{Appendix}\label{sec:appendix}
In \cite{alanwar2021data}, two approaches were proposed to estimate a reachable set for the state, in which a set $\tilde{\mathcal{R}} \subset \mathbb{R}^{n_x}$ such that $x(t)\in\tilde{\mathcal{R}}$ for all $t \in \mathbb{N}$ is constructed using the offline data and updated on-the-fly to yield a more accurate set $\hat{\mathcal{R}}_t$ using online measurements.
The two approaches are referred to as  reverse mapping and implicit intersection.
The main idea of the former is to explicitly recover the state from the output based on the SVD of matrix $C$.
However, since $C$ is unknown in this paper, the first approach cannot be employed directly.
Hence, this section focuses on extending the implicit intersection approach to fit our setting in Assumption \ref{as:iso-io} using zonotopes.

Assumption \ref{as:noise} is restated in terms of zonotopes as follows. Noise $w(t) \in \mathcal{Z}_w$ and $v(t) \in \mathcal{Z}_v$ are assumed to belong to the bounding zonotopes $\mathcal{Z}_w = \langle c_w, G_w \rangle \subset \mathbb{R}^{n_x}$ and $\mathcal{Z}_v = \langle c_v, G_v \rangle \subset \mathbb{R}^{n_y}$, respectively.
Define the matrix zonotope of pair $(A,B)$ as
\begin{equation}
	\mathcal{N}_{\Sigma} = \{ [A~ B]| X_1 = AX_0 + BU_0 + W_0, \ W_0 \in \mathcal{M}_w\}
\end{equation}
where $\mathcal{M}_w = \langle C_{\mathcal{M},w}, G_{\mathcal{M},w}^{(1:\xi T)}\rangle$ is the zonotope of process noise sequence $W_0$.
Similarly, the matrix zonotope of $C$, i.e., $\mathcal{M}_C = \langle C_{\mathcal{M},C}, G_{\mathcal{M},C}^{(1:\xi T)}\rangle$, is given by
\begin{align*}
	\mathcal{M}_C = (Y_0 \!- V_0)X_0^\dag \!= \Big(Y_0 - C_{\mathcal{M},v} \!- \sum_{i  = 1}^{\xi_{\mathcal{M},v}}\beta^{(i)}_{\mathcal{M},v} G_{\mathcal{M},v}\Big)X_0^\dag
\end{align*}
where $X_0^{\dag}$ exists under the condition \eqref{eq:rank:phi}, and $\mathcal{M}_v = \langle C_{\mathcal{M},v}, G_{\mathcal{M},v}^{(1:\xi T)}\rangle$ is the matrix zonotope of measurement noise $V_0$.
The set $\tilde{R}_t=\langle\tilde{c}_t, \tilde{G}_t\rangle$ over-approximates the exact reachable set with $\tilde{R}_{t + 1} = \mathcal{M}_{\Sigma}(\tilde{R}_t \times \mathcal{U}_t) + \mathcal{Z}_w$ and
\begin{equation*}
	\mathcal{M}_{\Sigma} = (X_1 - \mathcal{M}_w)
	\left[
	\begin{matrix}
		X_0\\
		U_0
	\end{matrix}
	\right]^{\dag}.
\end{equation*}
The true state can be expressed as $x_t = \tilde{c}_t + \tilde{G}_t z_t$ for some weight vector $z_t$.
Proposition 2 in \cite{alanwar2021data} can be restated as follows.
\begin{proposition}
	The intersection of $\tilde{R}_t$ and the region for $x(t)$ corresponding to $y(t)$ with noise $v(t) \in \mathcal{Z}_v$ satisfying \eqref{eq:sys} can be over-approximated by the zonotope $\hat{R}_t = \langle \hat{c}_t, \hat{G}_t \rangle$ with 
	\begin{align*}
		\hat{c}_t &= \tilde{c}_t + \lambda_t^*(y(t) - C_{\mathcal{M}_C} - C_{v,i})\\
		\hat{G}_t &= \!\big[
		\begin{matrix}
			(I \!- \lambda_t^* C_{\mathcal{M}_C})\tilde{G}_t&-\!\lambda_t^* G_{\mathcal{M}_C}&-\!\lambda_t^*G_{v}&-\!\lambda_t^* G_{\mathcal{M}_C}\tilde{G}_t
		\end{matrix}
		\big]
	\end{align*}
	where $\lambda_t^*$ is the optimal solution of the following problem
	\begin{equation*}
		\min_{\lambda_t}\big\Vert\! \big[
		\begin{matrix}
			(\!I \!-\! \lambda_t C_{\!\mathcal{M}_C})\tilde{G}_t&\!\!-\!\lambda_t G_{\!\mathcal{M}_C}&\!-\!\lambda_tG_{v}&\!-\!\lambda_t G_{\!\mathcal{M}_C}\!\tilde{G}_t	
		\end{matrix}
		\big]\! \big\Vert_F^2
	\end{equation*}
	and $\Vert \cdot \Vert_F^2$ denotes the matrix Frobenius norm.
\end{proposition}
Based on the above Proposition, the state estimator can be obtained leveraging \cite[Algorithm 1]{alanwar2021data}. 
\bibliographystyle{IEEEtranS}
\bibliography{ddest}

\end{document}